\documentclass[11pt,oneside,reqno]{amsart}

\usepackage{graphicx}
\usepackage{pict2e}
\usepackage{amssymb}
\usepackage{amsthm}                
\usepackage[margin=1in]{geometry}  
\usepackage{graphics}
\usepackage{epstopdf}
\usepackage{amsmath}
\usepackage{graphicx}
\usepackage{subfigure}
\usepackage{latexsym}
\usepackage{amsmath}
\usepackage{amsfonts}
\usepackage{amssymb}
\usepackage{tikz}
\usepackage{mathdots}
\usepackage{comment}
\usepackage{float}
\usepackage[mathscr]{euscript}
\usepackage{enumerate}
\usepackage{epsfig}
\usepackage { graphicx }
\usepackage { hyperref }
\usepackage { graphics }
\usepackage { graphicx }
\usepackage{xparse}
\usepackage{epstopdf}
\usepackage {eufrak}
\usepackage[utf8]{inputenc}
\usepackage{longtable}

\usepackage[lite]{amsrefs}

\theoremstyle{definition}
\newtheorem{theorem}{Theorem}[section]
\newtheorem{lemma}[theorem]{Lemma}

\theoremstyle{definition}

\theoremstyle{remark}
\newtheorem{remark}[theorem]{Remark}
\numberwithin{equation}{section}

\numberwithin{equation}{section}

\begin{document}


\title{Segmenting a Surface Mesh into Pants Using Morse Theory}


\author{Mustafa Hajij}
\address{Department of Mathematics, University of South Florida}
\email{mhajij@usf.edu}

\author{Tamal Dey}
\address{Department of Computer Science, Ohio State University}
\email{tamaldey@cse.ohio-state.edu}

\author{Xin Li}
\address{Department of Electrical Engineering, Louisiana State University}
\email{xinli@cct.lsu.edu }







\maketitle
\begin{abstract}
A pair of pants is a genus zero orientable surface with three boundary components.
A pants decomposition of a surface is a finite collection of 
unordered pairwise disjoint simple closed curves embedded in the surface 
that decompose the surface into pants. In this paper we present two
Morse theory based algorithms for pants decomposition of a surface mesh. 
Both algorithms operates on a choice of an appropriate Morse function on the
surface. The first algorithm uses this Morse function to 
identify \textit{handles} that
are glued systematically to obtain a pant 
decomposition. The second algorithm
uses the Reeb graph of the Morse function
to obtain a pant decomposition. Both algorithms work for surfaces
with or without boundaries. Our preliminary implementation of the two
algorithms shows that both algorithms run in much less time than an 
existing state-of-the-art method, and the Reeb graph based algorithm
achieves the best time efficiency. Finally, we demonstrate the robustness of our algorithms against noise.

\end{abstract}

\begin{figure}[h]\centering
  \includegraphics[width=0.51\textwidth]{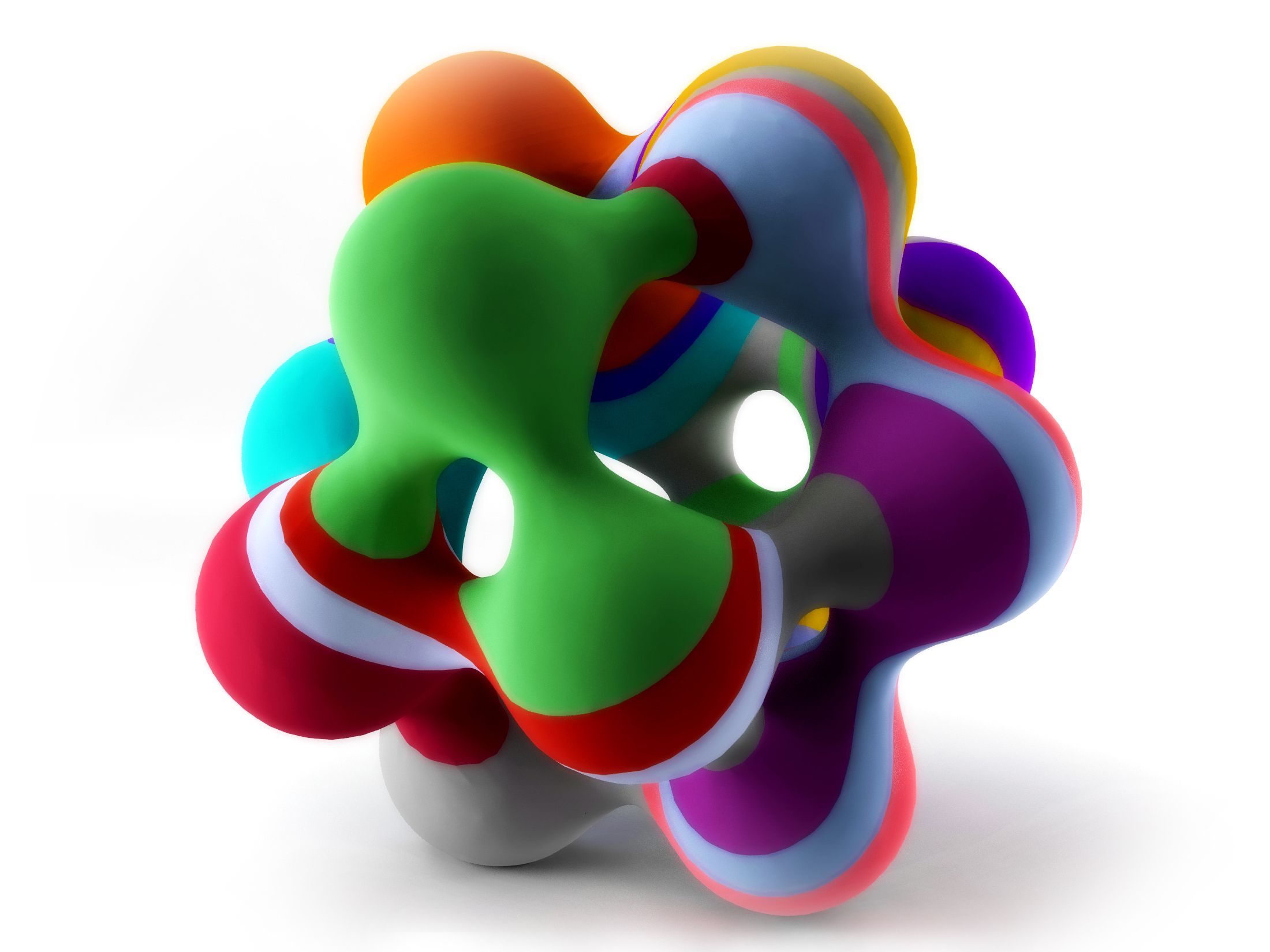}
  \caption{A pants decomposition of a high genus surface by our algorithm.}
  \label{fig:square}
\end{figure}



\section{Introduction}


Segmenting a surface mesh into simple pieces for further processing 
is a fundamental problem in many mesh processing applications
such as texture mapping~\cite{levy2002least}, collision detection~\cite{li2001decomposing}, skeletonization~\cite{biasotti2003overview} and three-dimensional shape retrieval~\cite{zuckerberger2002polyhedral}.
In surface matching, authors
in \cite{li2009surface} showed the use of a particular type of 
mesh segmentation called
a {\em pants decomposition}, see Figure~\ref{fig:square}. It
derives its name from its constituent piece termed
a {\em pair of pants}, which, up to topology, is a genus zero orientable surface with 
$3$ boundary components. Besides surface matching, pants decomposition has found applications in  surface classification and indexing \cite{jin2009computing}, and consistent mesh parametrization \cite{kwok2012constructing}. This type of segmentation can be viewed as a \textit{common base domain segmentation} \cite{kwok2012constructing}, where one partitions the mesh into parts with a common property such as having the same topology. Authors in 
\cite{li2009surface} presented a
pants decomposition algorithm based on the computations of certain basis cycles 
in the first homology group of the input surface.
Other pants decomposition related algorithms proposed in the graphics literature can be found in ~\cite{zhang2012optimizing} where the authors enumerates different classes of pants decompositions. However, these methods rely on computing certain curves on the surface called \textit{Handle and Tunnel loops} ~\cite{dey2007computing,dey2013efficient}. Computing such curves is expensive and more importantly requires the surface to be embedded in $\mathbb{R}^3$ and not have any boundary.  Morse theory allows our algorithms to get rid of these two constraints. Moreover, pants decomposition relies on finding a collection curves on the surface with certain topological properties. Finding these curves and manipulating them is a difficult problem. Morse theory allows us to avoid defining these curves explicitly by realizing them implicitly as level curves of a Morse function.

Many segmentation algorithms have been proposed in the graphics literature includes~\cite{mangan1999partitioning, shamir2008survey,li2009surface,chen2009benchmark,shlafman2002metamorphosis,attene2006hierarchical,yu2015geometry,li2016geometry}. One reason for the variety of segmentation algorithms suggested in the literature is that there is no one universal good algorithm that suits all applications. The techniques used in the algorithms are related to other areas in computer graphics such as image segmentation~\cite{reynolds1995robust,tomasi1998bilateral} and machine learning~\cite{karypis1998software,cover1967nearest}. For good surveys on various segmentation algorithms see~\cite{attene2006mesh,shamir2008survey}. 

In this work we propose two algorithms based on Morse theory for computing 
a pants decomposition of an input surface mesh. The Morse theory 
connects 
the geometry and topology of manifolds via
the critical values/points of a specific class of real-valued functions called
Morse functions. Intuitively, given such a function $f$ on a manifold $M$, 
Morse theory studies the topological changes of the level sets of $f$. 
These level sets change in topology only at the critical values. 
Consequently, the space in between two consecutive critical levels becomes
a product space of a fiber with an interval. This fact is used to decompose a surface into pants and cylinders, and the latter parts are glued inductively into pants.
The second algorithm exploits a reduced structure called {\em Reeb graph} ~\cite{reeb1946points}
derived from a function $f$ on the surface. The Reeb graph of $f$
is a quotient space derived from $M$ and $f$. For a Morse
function, all vertices in the Reeb graph have 
valence $3$ except the
extrema. We exploit this property to compute a pair of pant for every
degree-3 vertex of the Reeb graph and treat the extrema specially. 

Morse theory, in its original form, assumes smooth settings ~\cite{milnor1963morse}. 
For operating on surface meshes, we need a piecewise linear (PL)
version of the theory. 
We make this transition using the PL Morse theory
proposed by ~\cite{critical1967}. 
We show how one can compute a function on a surface mesh satisfying the
specific property that our algorithms require. Furthermore, we extend
our basic algorithms to the case when the surface has boundaries, or when
the function $f$ allows degenerate critical points such as monkey saddles.

Our experiments show that both of our algorithms run much faster in practice
than the algorithm of Li et al. while the Reeb-graph based algorithm performs the best. Furthermore, our tests  show that both of the algorithms presented here are robust to noise.


\section{Morse Theory and Handle Decomposition for Surfaces}
Let $M$ be a compact smooth surface, and let $I=[a,b]\subseteq \mathbb{R} $, where $a<b$, be a closed interval. Let $f:M \longrightarrow I$ be a smooth function defined on $M$. A point $x \in M$ is called a \textit{critical point} of $f$ if the differential $df_x$ is zero. A value $c$ in $\mathbb{R}$ is called a \textit{critical value} of $f$ is $f^{-1}(c)$ contains a critical point of $f$. A point in $M$ is called a \textit{regular point} if it is not a critical point. Similarly, a value $c \in \mathbb{R}$ is called regular if it is not critical. The inverse function theorem implies that for every regular value $c$ in $\mathbb{R}$ the level set $f^{-1}(c)$ is a $1$-manifold, i.e., $f^{-1}(c)$ is a disjoint union of simple closed curves. A critical point is called \textit{non-degenerate} if the matrix of the second partial derivatives of $f$, called the \textit{Hessian matrix}, is non-singular. 

The definition of a Morse function is motivated mainly by the following Lemma.
\begin{lemma}
(Morse Lemma) Let $M$ be a smooth surface, $f:M \longrightarrow I$ be a smooth function and $p$ be a non-degenerate critical point of $f$. We can choose a chart $(\phi,U)$ around $p$ such that $f \circ \phi^{-1}$ takes exactly one of the following three forms:
\\
\begin{enumerate}

\item $f \circ \phi^{-1}(X,Y)= X^2+Y^2+c$.
\item $f \circ \phi^{-1}(X,Y)= -X^2-Y^2+c$.
\item $f \circ \phi^{-1}(X,Y)= X^2-Y^2+c$.
\end{enumerate}
\end{lemma}

The \textit{index} of a critical point $x$ of $f$, denoted by $index_f(x)$, is defined to be the number of negative eigenvalues of its Hessian matrix. Since the Hessian of a scalar function on smooth surface is a $2\times 2$ symmetric matrix, then the index takes the values $0,1$ or $2$.
One can see that on a non-degenerate critical point of index $0$ the function $f$ takes a minimum value, on a non-degenerate critical point of index $1$, the graph of the function looks like a saddle and on a non-degenerate critical point of index $2$ the function $f$ takes a maximum value. See Figure \ref{minmaxsaddle}

\begin{figure}[h]
  \centering
   {\includegraphics[width=0.55\textwidth]{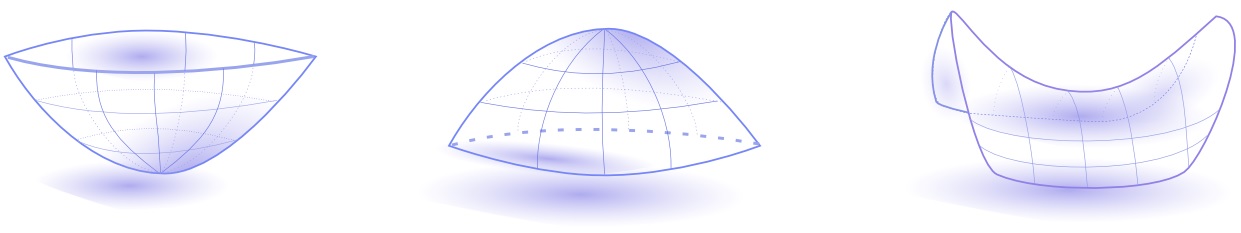}
   \caption{Minimum, Maximum, and Saddle.}
    \label{minmaxsaddle}
 }
\end{figure} 

If all the critical points of $f$ are non-degenerate and all critical points have distinct values, then $f$ is called a \textit{Morse function}. If the surface $M$ has boundary, then we also require two other conditions : (1) $f^{-1}(\partial I) = \partial M$
and (2) the critical points of $f$ lie in the interior of $M$.

\subsection{Handle Decomposition of a Surface}
Our first algorithm uses the attachment of cylinders called handles at
the critical points. We need a few definitions first.\\ 

 Let $M$ be a smooth surface and let $f:M \longrightarrow \mathbb{R}$ be a Morse function defined on $M$. Define the set 
\begin{equation*}
M_{f,t} =\{ x\in M :f(x)\leq t \}.
\end{equation*} 
Let $a,b$, $a<b$, be two reals. Define 
\begin{equation*}
M_{f,[a,b]} =\{ x\in M : a\leq f(x)\leq b \}.
\end{equation*} 
When it is clear from the context, we will drop $f$ from the notation and use simply $M_{t}$ and $M_{[a,b]}$ to refer to the previous two sets.
Morse theory studies the topological changes of $M_{t}$ as $t$ varies. The following is well-known~\cite{matsumoto2002introduction}.
\begin{theorem}
Let $f:M\longrightarrow \mathbb{R}$ be a smooth function on a smooth surface $M$. For two reals $a,b$, $a<b$, if $f$ has no critical values in the interval $[a,b]$, then the surfaces $M_{a}$ and $M_{b}$ are diffeomorphic. 
\end{theorem}
The previous theorem says that the topology of the surface $M_t$ does not change as $t$ passes through regular values. In the following we use $D^1$ to denote the interval $[0,1]$. The end points of $D^1$ are given by $\partial D^1=\{0,1\}$. Given a Morse function $f$ on M, the following theorem gives s precise description for the change that occurs in the topology of $M_t$ as $t$ passes through a critical value. 
\begin{theorem}
\label{main thm}
Let $f:M \longrightarrow \mathbb{R}$ be Morse function. Let $p$ be a critical point of index $i$ and $f(p)=t$ be its corresponding critical value. Let $\epsilon$ be chosen small enough so that $f$ has no critical values in the interval $[t-\epsilon,t+ \epsilon ]$.\\ 
\begin{enumerate}
\item
 If $index_f(p)=0$, then $M_{t+\epsilon}$ is diffeomorphic to the disjoint union of $M_{t-\epsilon}$ and a $2$-dimensional disk $D^2$.
\item If $index_f(p)=1$, then $M_{t+\epsilon}$ can be obtained from $M_{t-\epsilon}$ by attaching a $1$-handle. This means that $M_{t+\epsilon}$ can be obtained by gluing a rectangular strip $D^1\times D^1$ to the boundary of $M_{t-\epsilon}$ along $D^1\times \partial D^1 $.
\item If $index_f(p)=2$, then $M_{t+\epsilon}$ can be obtained by capping off the surface $M_{t-\epsilon}$ with a disk $D^2$. This means that $M_{t+\epsilon}$ is obtained by gluing a disk $D^2$ along its boundary $\partial D^2$ to one of the boundary components of $M_{t-\epsilon}$.  
\end{enumerate}
\end{theorem}
\subsection{Morse Theory for Triangulated Surfaces}

Morse theory was extended to triangulated surfaces ($2$-manifolds)
by Banchoff~\cite{critical1967}. Let $M$ be a triangulated surface, 
and let $f :M\longrightarrow I$ be a 
piece-wise linear (PL) continuous function on $M$. 
The link $Lk(v)$ of a vertex $v$ is defined as the set of all vertices $w$ that $v$ shares an edge $[v,w]$ with.
The upper link of $v$ is defined as 
\begin{equation*}
Lk^+(v)=\{ u \in Lk(v):   f(u)>f(v)\},
\end{equation*}
and the lower link is defined by
\begin{equation*}
Lk^{-}(v)=\{u \in Lk(v):  f(u)<f(v)\},
\end{equation*}
and mixed link 
\begin{equation*}
Lk^{\pm}(v)=\{ (u_1,u_2) : f(u_1)< f(v) < f(u_2)\}.
\end{equation*}
\begin{figure}[H]
\centering
  \includegraphics[width=0.5\textwidth]{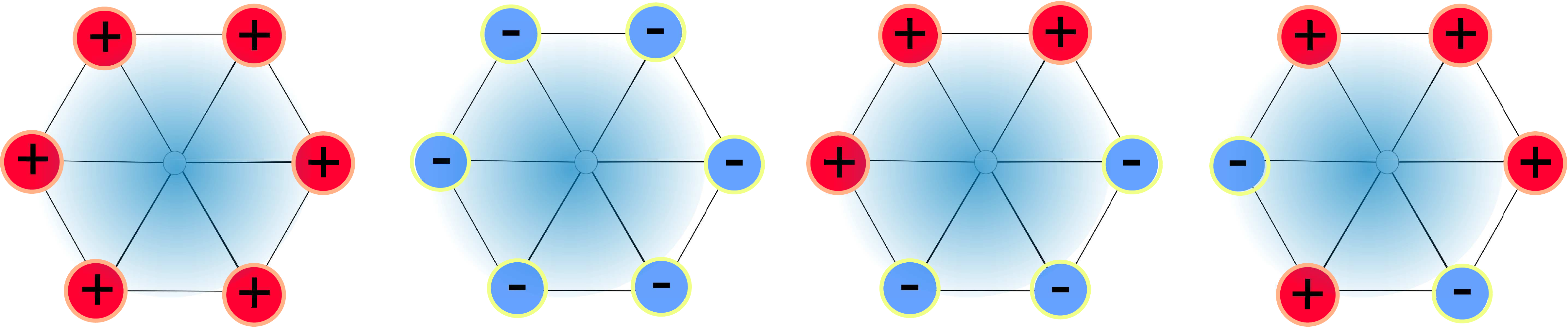}
  \caption{(a) minimum, (b) maximum, (c) regular vertex, (d) saddle .}
  \label{classification}
\end{figure}
Using the link definitions we can classify the interior vertices of $M$. The
boundary vertices are treated separately as we will require $f$ to have
identical values on them. 
An interior vertex $v$ is \textit{regular} if $|Lk^{\pm}(v)|=2$, is a
\textit{maximum} with index $2$ if $|Lk^+(v)|=0$, is
a \textit{minimum} with index $0$ if $|Lk^-(v)|=0$, 
and is a \textit{saddle} with index $1$ and multiplicity $m\geq 1$ 
if $|Lk^{\pm(v)}|=2+2m$. 
See Figure \ref{classification}. A vertex is 
a simple critical point if it is either minimum, or maximum, or a saddle
with multiplicity $1$. A PL function on a closed triangulated surface $M$ is  
\textit{PL Morse } if all its vertices are either PL 
regular or simple PL critical and have distinct function values. If the mesh $M$ has a non-empty boundary then we also require (1) $f^{-1}(\partial I) = \partial M$, 
and (2) the critical vertices of $f$ lie in the interior of $M$.

\subsection{Reeb Graph}
Let $M$ be a surface possibly with boundary $\partial M$, and let $f:M \longrightarrow [0,1]$ be continuous. Define the equivalence relation $\sim$ on $M$ by $x \sim y$ if and only if $f(x)=f(y)=c\in [0,1]$ and
$x$ and $y$ belong to the same connected component of the level set $f^{-1}(c)$. The set $R(f)=X/{\sim}$ with the standard quotient topology is called \textit{Reeb graph} of $f$. When $f$ is smooth and Morse, or PL Morse,
every vertex of the $R(f)$ arises from a critical point of $f$
or a boundary component. Every maximum or minimum of $f$ 
gives rise to a degree $1$-node of $R(f)$. Since $f$ is Morse,
every boundary component also gives rise to a degree $1$-node and 
every saddle of $f$ gives rise to degree $3$-node. If $M$ is an embedded surface without boundary then $M$ can be recovered up to a homeomorphism from $R(f)$ as the boundary of an oriented $3$-dimensional regular neighborhood of the graph $R(f)$. See Figure \ref{Reebexample} for an example of a Reeb graph.
\begin{figure}[h]
  \centering
   {\includegraphics[width=0.5\textwidth]{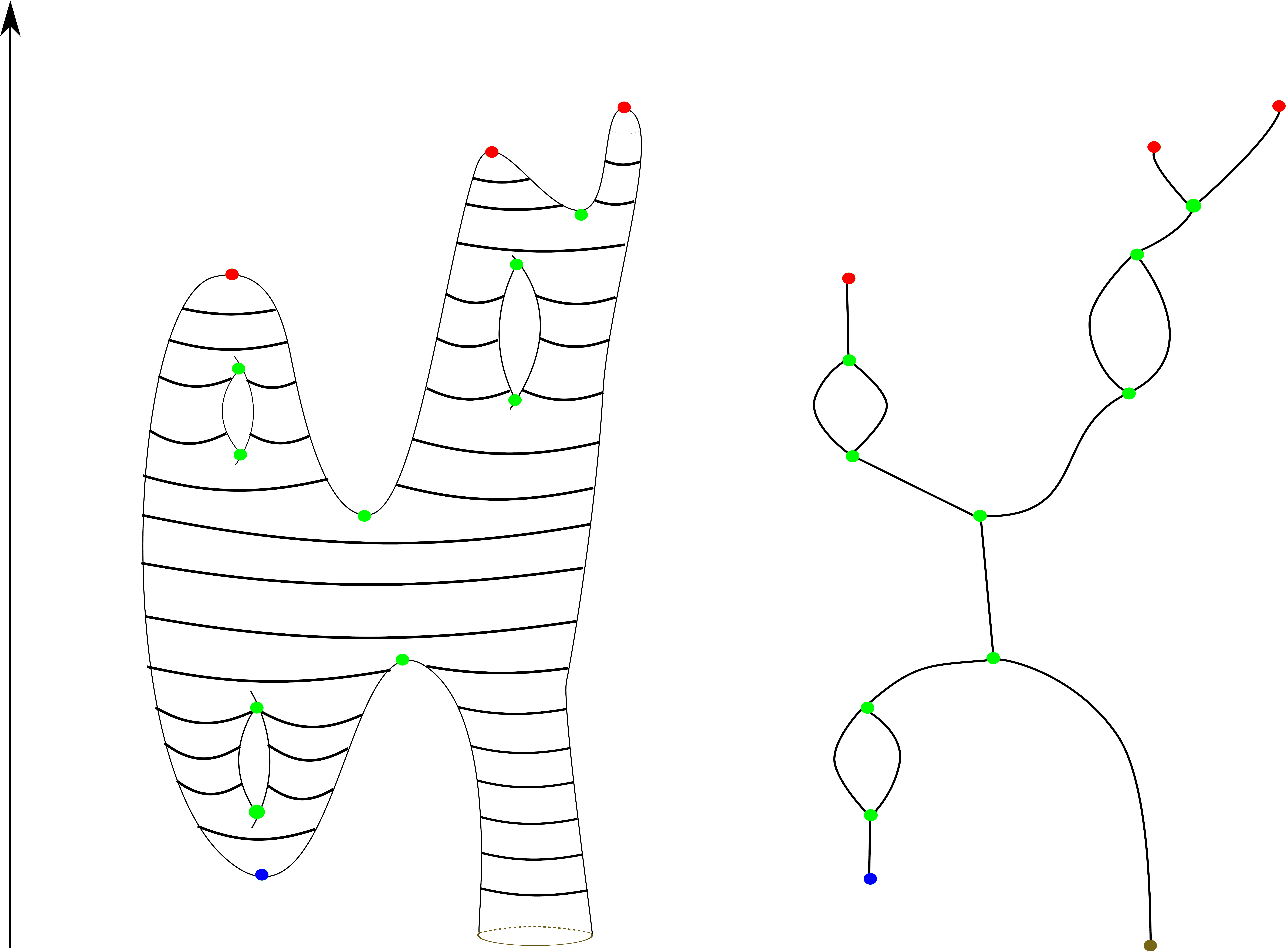}
  }
  \put(-240,165){$f$}
  \caption{An example of a Reeb graph.}
  \label{Reebexample}
\end{figure}
\begin{remark}
One reason that we require the Morse function to assume constant 
values on boundaries
is that it makes it consistent with the definition of the Reeb graph we give here. Note that each boundary component maps exactly to one point on the Reeb graph. 
\end{remark} 

The definition of Reeb graph goes back to G. Reeb~\cite{reeb1946points}. It was first introduced to computer graphics in~\cite{sato1994matrix}. Reeb graph has found applications in shape understanding~\cite{attene2003shape}, quadrangulation~\cite{hetroy2003topological}, surface understanding~\cite{biasotti2000extended}, segmentation~\cite{werghi2006functional}, parametrization~\cite{patane2004graph,zhang2005feature}, animation~\cite{kanongchaiyos2000articulated} and many other applications. Reeb graphs algorithms can be found in many papers such as~\cite{shinagawa1991constructing,cole2003loops,pascucci2007robust,doraiswamy2009efficient}. The most efficient algorithms in terms of time complexity are due to
~\cite{Salman12,HWW10}.

\section{Pants Decomposition}
Let $M$ be a compact, orientable, and connected surface. We say that $M$ is of type $(g,b)$ if $M$ is of genus $g$ and has a $b$ boundary components. A pair of pants is a surface of type $(0,3)$. See Figure \ref{pant}.

\begin{figure}[h]\centering
  \includegraphics[width=0.3\textwidth]{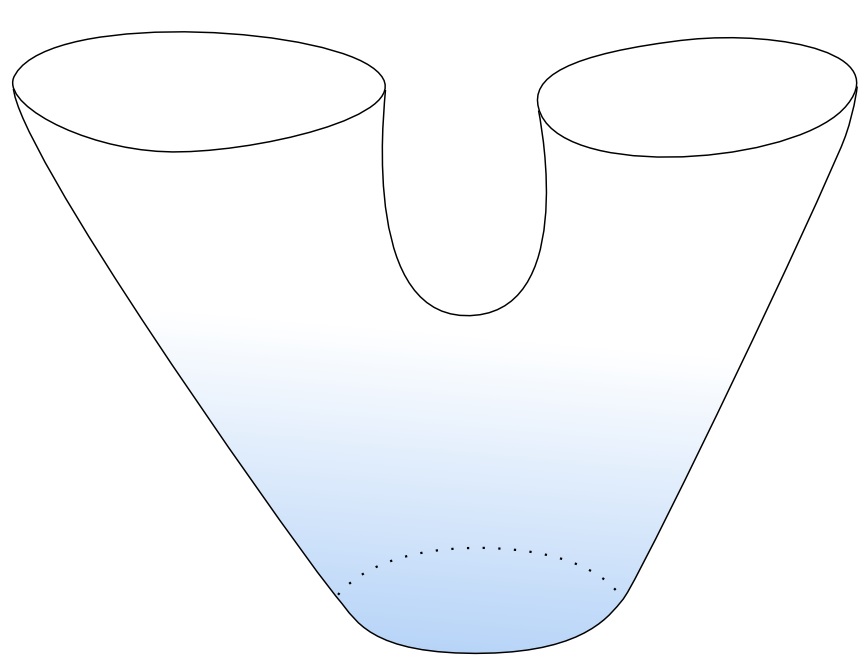}
  \caption{A pair of pants}
  \label{pant}
\end{figure}

 A pants decomposition of $M$ is a finite collection of unordered pairwise disjoint simple closed curves $\{c_1,...c_n\}$ embedded in $M$ with the property that the closure of each connected component of
$M-(c_1\cup...\cup c_n)$ is a pair of pants. Two pants decompositions of $M$ are equivalent if they are isotopic. See Figure \ref{two} for an example of two non-isotopic pants decompositions of a genus-$2$ surface.
\begin{figure}[h]\centering
  \includegraphics[width=0.8\textwidth]{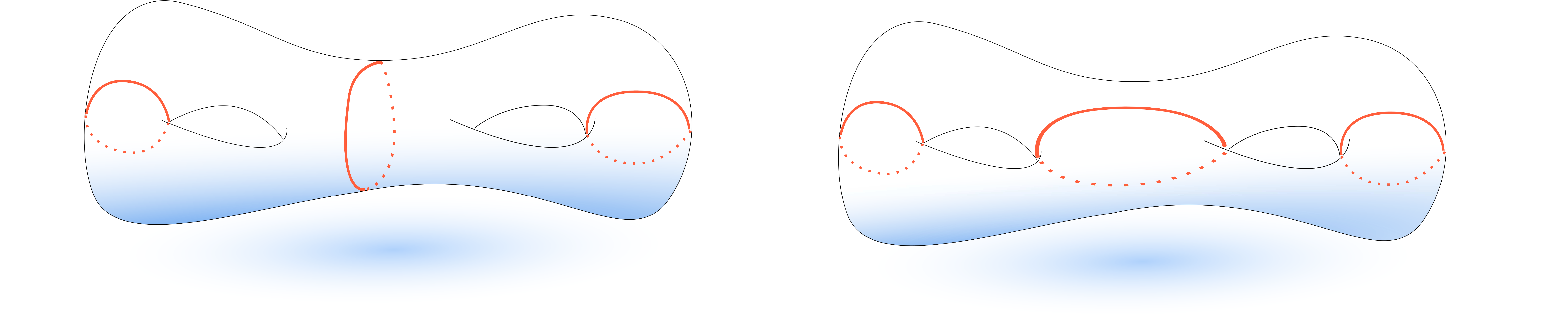}
  \caption{Two non-isotopic pants decompositions of a genus $2$ surface. }
  \label{two}
\end{figure}

Let $M$ be a compact orientable surface and connected surface of type $(g,b)$. The Euler characteristic
of M, denoted by $\chi(M)$ is defined as $\chi(M)=2-2g-b$.
Every connected, compact and orientable surface $M$ with $\chi(M)<0$, genus $g$ and $b$ boundary components admits a pants decomposition with $3g-3+b$ simple closed curves and the number of complementary components is $2g-2+b=|\chi(M)|$. 
\section{Handle-Based Pants Decomposition}
In this section we use handles given by a Morse function to design an algorithm for decomposing a surface $M$ with $\chi(M)<0$ into a collection of surfaces of type $(0,3)$. Our algorithm works for arbitrary surface $M$ with $\chi(M)<0$ 
and with or without a boundary. However, in order to guarantee the correctness of our algorithm we must choose a function with certain properties. 
Ideally, this function should be PL Morse with these properties
when $M$ is a surface mesh.
We cannot always guarantee that the function 
is PL Morse, but we can compute one that
satisfies all required properties except
the simplicity of the saddles.
We first describe the algorithm assuming a PL Morse function and later
mention how to handle the exceptions of degenerate saddles.


\subsection{Orientable Surfaces With $\chi( \cdot )<0$ and No Boundary}
\label{alg1}
In this section we give an algorithm to compute a pants decomposition of a 
triangulated surface with genus $g\geq 2$ without boundary. Let $M$ be a compact connected orientable surface with genus $g\geq 2$ without boundary and let $f$ be a PL Morse function on $M$. Suppose that $t_1,t_2,...,t_n$ are the critical values for $f$ ordered in an ascending order. Let $p_1,p_2,...,p_n$ be the corresponding critical points of $f$. Choose a real number $\epsilon > 0$ small enough so that for each $1\leq i \leq n$ there are no critical values for $f$ on the interval $[t_i-\epsilon,t_i+ \epsilon ]$ except $t_i$. Finally we assume that function $f$ has exactly one minimum and exactly one maximum. It is clear from the choice of the function $f$ that one of the points $p_1$ and $p_n$ is the global maximum and one of them is the global minimum. Since multiplying any Morse function on a surface by $-1$ changes its critical points of index $0$ to critical points of index $2$ and vise versa, we can always choose our Morse function $f$ so that $p_1$ is the global minimum and $p_n$ is the global maximum. It should be noted that such a function can be constructed in practice and we will talk about the construction of such functions later. We need the following lemma for the correctness of our algorithm. 
\begin{lemma}
\label{main lemma}
Let $M$ be a compact connected orientable surface with a Morse function $f$ chosen as specified above then $M_{t_3+\epsilon}$  is homeomorphic to a surface of type $(1,1)$ or a surface of type $(0,3)$. 
\end{lemma}
\begin{proof}
The choice of the scalar function $f$ implies immediately that for each  $ 
2 \leq i \leq n-1 $ we have $index_f( p_i)=1$. Moreover, by construction we have $index_f(p_1)=0$ and $index_f(p_n)=2$. By Theorem \ref{main thm} we conclude that $M_{t_1+\epsilon}$ is diffeomorphic to a disk and $M_{t_2+\epsilon}$ is diffeomorphic to surface of type $(0,2)$.

\begin{figure}[h]\centering
  \includegraphics[width=0.55\textwidth]{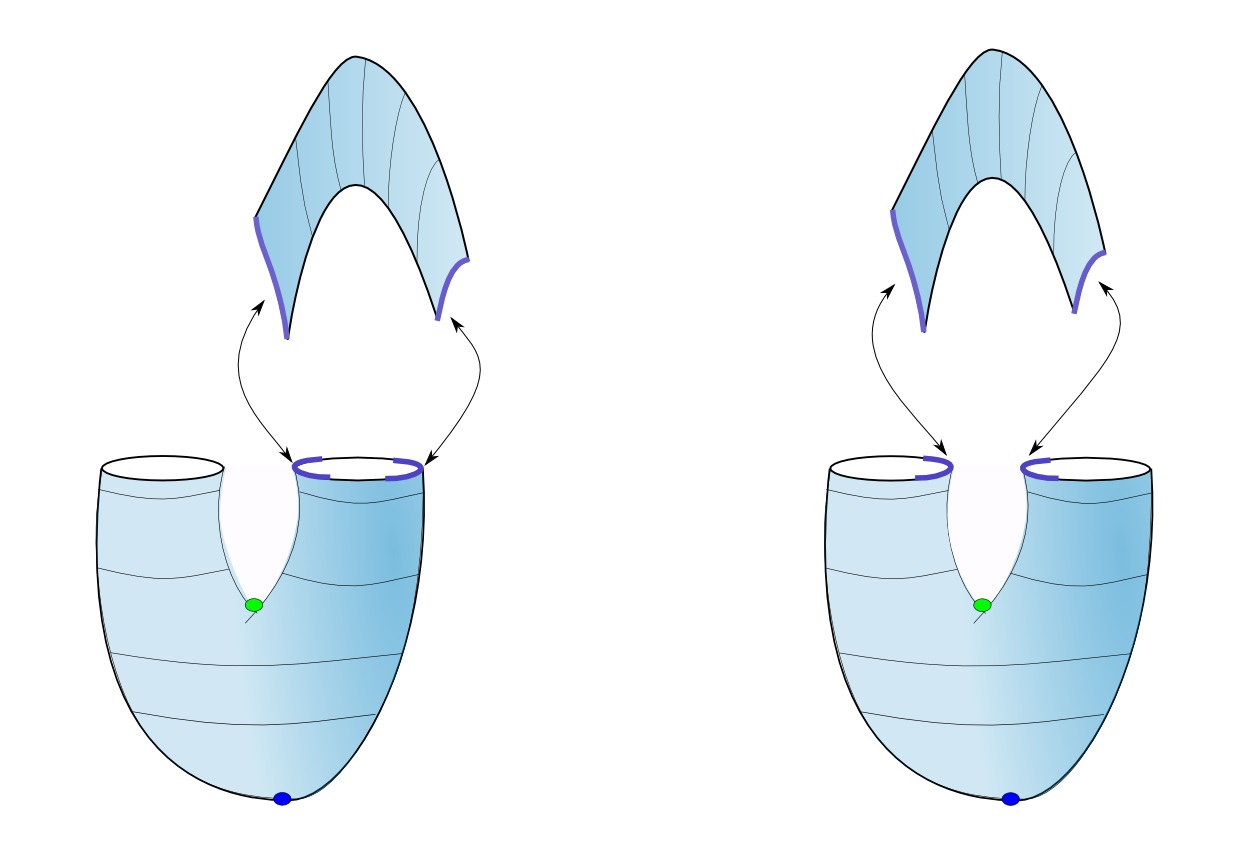}
  \small{
   \put(-247,125){$D^1\times \{0\}$}
    \put(-163,115){$D^1\times \{1\}$}
   }
  \caption{Two possible ways to glue a disk to the surface of type $(0,2)$.}
  \label{twopossibles}
\end{figure}

 Moreover, by Theorem \ref{main thm}, when $f$ passes through $t_3$ the surface $M_{t_3+\epsilon}$ is obtained from $M_{t_2+\epsilon}$ by gluing a rectangular strip $D^1\times D^1$ to the boundary of $M_{t_2+\epsilon}$ along $D^1\times \partial D^1$. Up to a homeomophism, there are two possible ways of gluing the rectangular strip $D^1\times D^1$ to the boundary of $M_{t_2+\epsilon}$ along $D^1\times \partial D^1$. See Figure \ref{twopossibles}. We either glue this rectangular strip to the same boundary component of $M_{t_2+\epsilon}$ to obtain a surface of type $(0,3)$ or we glue each side of the strip on one of the boundary components to obtain a surface of type $(1,1)$. See Figure \ref{initial}. 

\begin{figure}[h]\centering
  \includegraphics[width=0.38\textwidth]{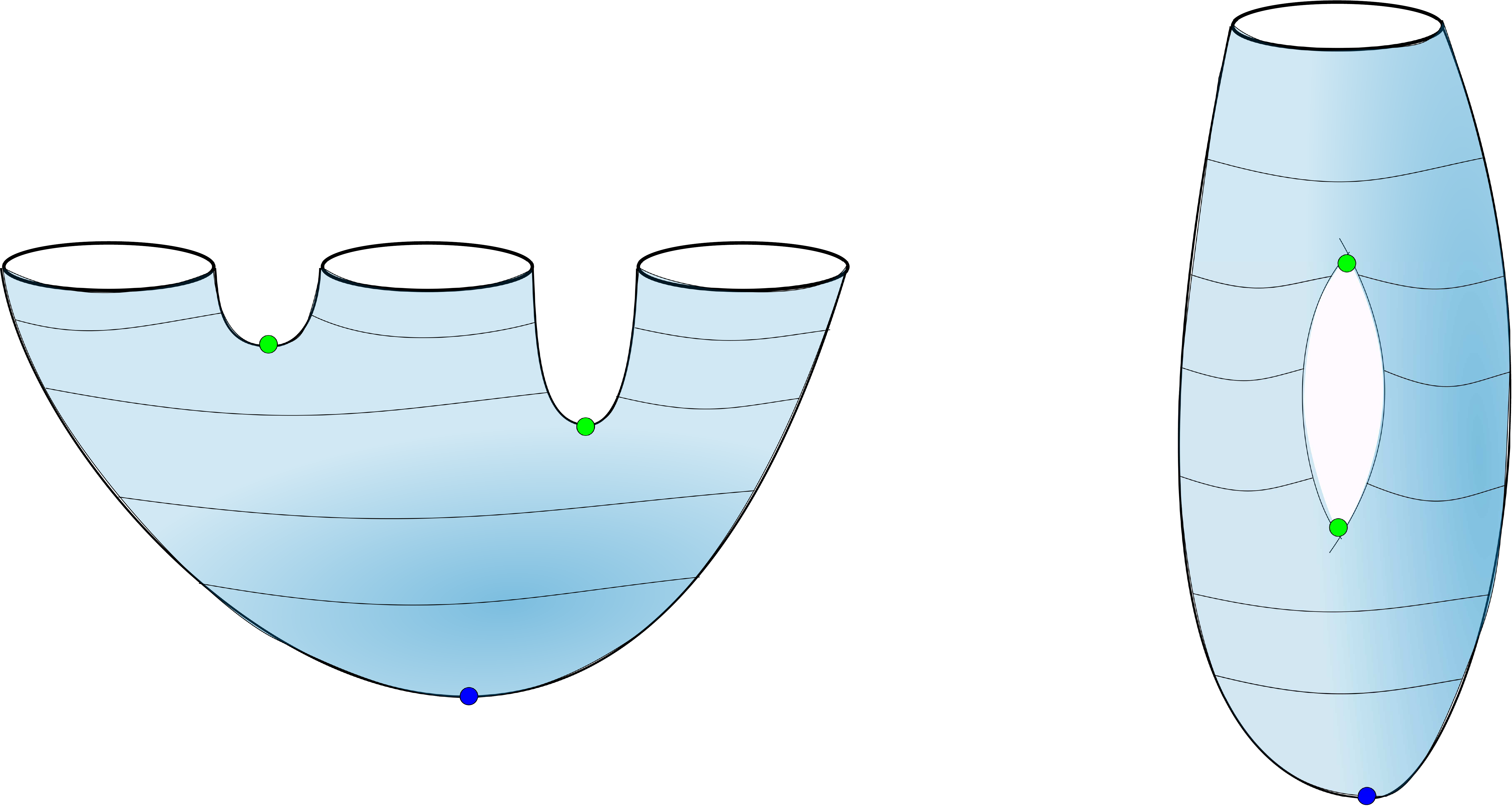}
  \caption{$M_{f,t_{3}+\epsilon}$  is homeomorphic to a type $(0,3)$ or to a type $(1,1)$}.
  \label{initial}
\end{figure}
\end{proof}

\begin{remark}
\label{main remark}
Using similar argument one can prove that the surface $M_{t_{n-2}-\epsilon}$ is either of type $(0,3)$ or a surface of type $(1,1)$.
\end{remark}
Lemma \ref{main lemma} holds in the case when the function $f$ is PL Morse on a triangulated surface. However, the case when the PL scalar function has saddle points with multiplicity larger than or equal to $2$ needs special treatment and Lemma \ref{main lemma} is no longer valid. We deal with such cases in section \ref{degenerate}.\\  

Lemma \ref{main lemma} and remark \ref{main remark} will be used to obtain the first and the last pants in our pants decomposition. The algorithm is as follows:

\begin{enumerate}
\item Compute the critical points of $f$ and put them in an ascending order. Let $p_1,p_2,...,p_n$ be the sequence of ordered critical points of $f$ and let $t_1,t_2,...,t_n$ be their corresponding critical values. Note that $n=2g+2$ by our choice of the scalar function $f$.

\item For each $3\leq i < n-3$ let $c_i=\frac{t_i+t_{i+1}}{2}$. After reindexing, we define the set $C=\{c_i| 1 \leq i \leq 2g-3\}$. In other words, the set $C$ is a set of ordered regular values for $f$ such that there is exactly one critical value for the function $f$ in the intervals $[c_i,c_{i+1}]$ for $1 \leq i \leq 2g-3$. See Figure \ref{algorithm} for an example.

\begin{figure}[H]\centering
  \includegraphics[width=0.41\textwidth]{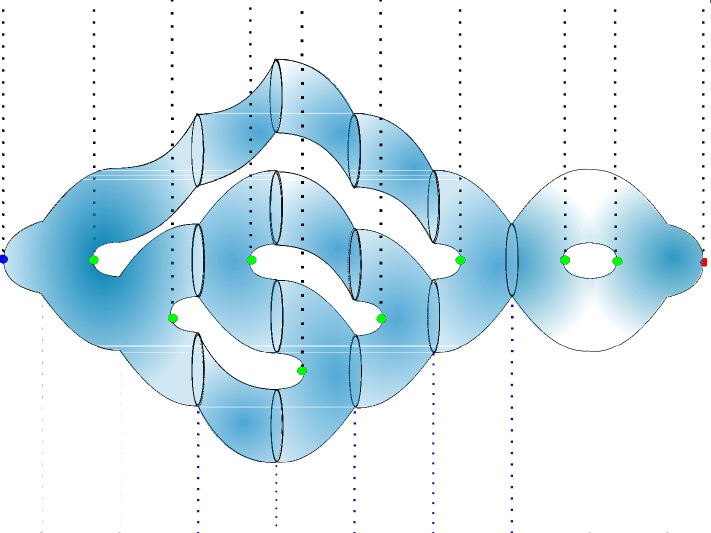}
  \put(-150,180){Increasing values of $f$}
   \put(-195,155){$t_1$}
    \put(-170,155){$t_2$}
     \put(-150,155){$t_3$}
     \put(-146,-7){$c_1$}
     \put(-132,155){$t_4$}
     \put(-125,-7){$c_2$}
     \put(-115,155){$t_5$}
     \put(-102,-7){$c_3$}
      \put(-93,155){$t_6$}
      \put(-83,-7){$c_4$}
       \put(-70,155){$t_7$}
       \put(-60,-7){$c_5$}
       \put(-46,155){$t_8$}
        \put(-31,155){$t_{9}$}
        \put(-10,155){$t_{10}$}
  \caption{Cutting a surface of genus $4$ along the values $c_i$.}
  \label{algorithm}
\end{figure}
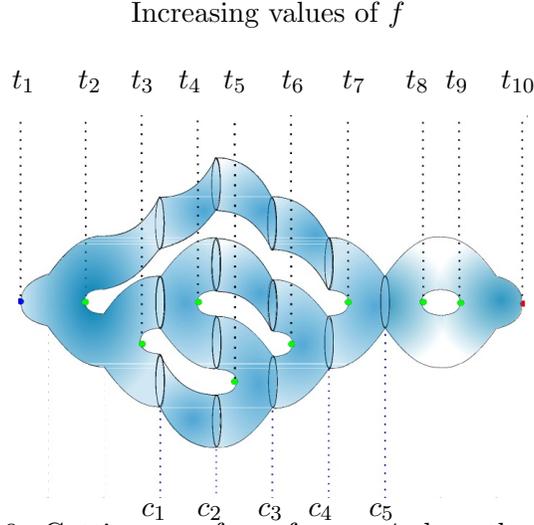

\item Cut the surface $M$ along the level sets $f^{-1}(c)$ for all $c\in C$. Note again that the values $c \in C$ are all regular values.

\item By Lemma \ref{main lemma} the surface $M_{c_1}$ is either of type $(1,1)$ or of type $(0,3)$. If $M_{c_1}$ is of type $(1,1)$, then we trace a loop from the saddle point $p_2$ to the minimum point $p_1$ and  cut the surface along that loop. We explain later a method of tracing this loop.
See Figure \ref{torus part}. Otherwise, if the surface is of type $(0,3)$ we do not need to do anything. In either cases, we denote the resulting surface of type $(0,3)$ by $M_{initial}$.
\begin{figure}[h]\centering
  \includegraphics[width=0.1\textwidth]{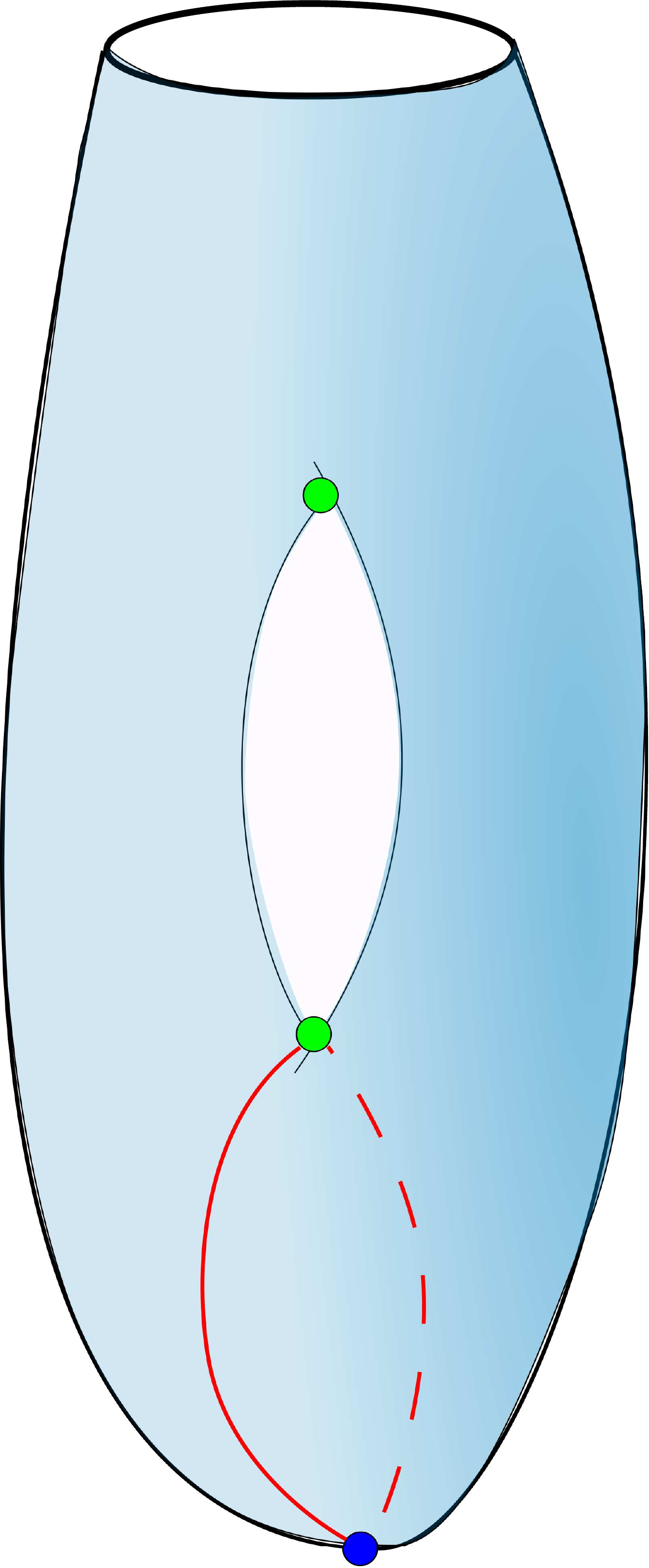}
  \caption{Tracing a loop from the first saddle point $p_2$ to the minimum $p_1$.}
  \label{torus part}
\end{figure}
\item Consider the manifold with boundary $M_{[c_1,c_{2}]}$. This surface is a finite disjoint union of one surface of type $(0,3)$ and multiple surfaces of type $(0,2)$. Attach every surface of type $(0,2)$ to $M_{initial}$. Note that this gluing does not change the homeomorphism type of $M_{initial}$. See Figure \ref{example} (b).

\item  For each $1< i < 2g-3$ consider the manifold with boundary $M_{[c_i,c_{i+1}]}$. This surface is again a finite disjoint union of one surface of type $(0,3)$ and multiple surfaces of type $(0,2)$. Attach every surface of type $(0,2)$ to $M_{[c_{i-1},c_{i}]}$. See Figure \ref{example} (c).
\begin{figure}[h]\centering
  \includegraphics[width=0.6\textwidth]{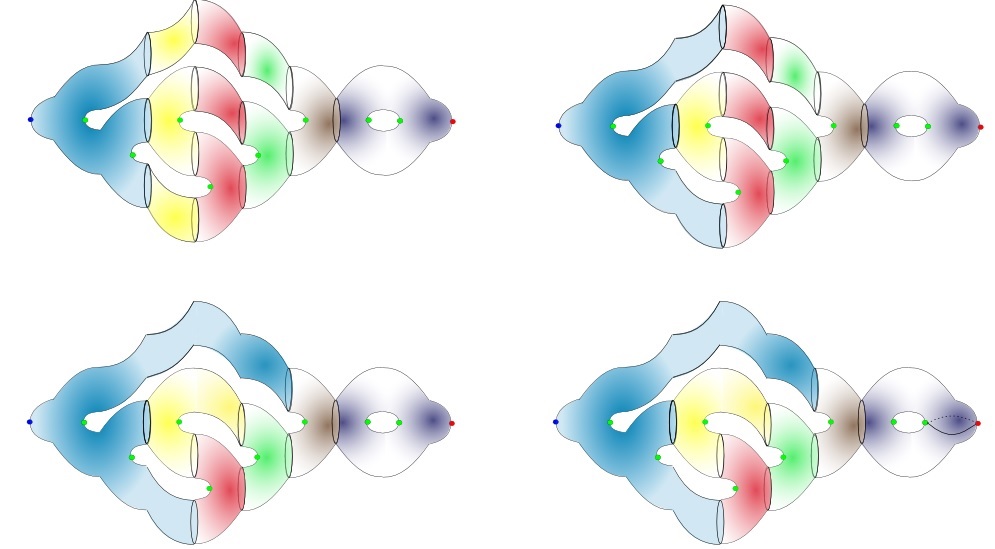}
   \put(-268,145){$(a)$}
   \put(-268,68){$(c)$}
   \put(-110,145){$(b)$}
     \put(-110,60){$(d)$}
 \caption{Illustration of the handle-based pants decomposition algorithm : (a) Cutting the surface along $c_1$, $c_2$, $c_3$, $c_4$ and $c_5$. Note that $M_{initial}$ is a pants in this example. (b) Attaching cylinders that appear on the second level to $M_{initial}$ and keeping the pant component. (c) Attaching cylinders at each level to the previous level and keeping the pant component. (d) The surface $M_{-f,c_{2g-3}}=M_{-f,c_4}$ is a surface of type $(1,1)$ so we trace a loop from the saddle to the minimum of $-f$ and cut the surface along this loop.} 
  \label{example}
\end{figure}
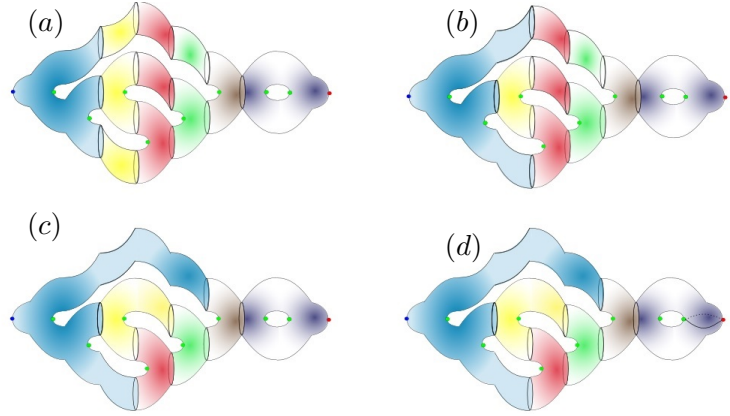
\item
 The remaining part $M_{-f,c_{2g-3}}$ is either of type $(1,1)$ or of type $(0,3)$. If the surface $M_{-f,c_{2g-3}}$ is of type $(1,1)$ then trace a loop from the saddle point $p_{n-1}$ to the point $p_n$ and cut the surface $M_{-f,c_{2g-3}}$ along this loop to obtain a pant. Otherwise, if the surface is of type $(0,3)$ we do not need to do anything.
 \end{enumerate}
Note that the algorithm constructs a collection of pants inductively. We start by having the first pair of pants in step $4$ and then we go to the next level which is a finite disjoint union of a single pant and some topological cylinders. We attach the cylinders to the previous pant and then we go to the next level and repeat the same process.

We should clarify here what we mean by tracing a loop from the saddle to the minimum point we mentioned in step $(4)$. Let $p_2$ the simple saddle point and $p_1$ be the unique minimum point for the PL scalar function $f$ on $M$ specified in our algorithm above. The point $p_2$ is a simple saddle, hence the set $Lk^-(p_2)$ can be decomposed into two disjoint connected components $A$ and $B$. Pick the vertex $v_A$ in $A$ such that $f(v_A)\leq f(v)$ for all $v \in A$. The vertex $v_B$ is chosen similarly. A \textit{descending path} from a regular vertex $v_0$, denoted by $dpath(v_0)$, is defined to be a finite sequence of vertices $\{v_0,...,v_k\}$ on $M$ such that $<v_i,v_{i+1}>$ is an edge on $M$ for $0\leq i \leq k-1$, $f(v_i)<f(v_{i-1})$, and $v_k$ is a minimum. A loop connecting $p_2$ and $p_1$ can be computed as the concatenation of $dpath(v_A)$, $<p_2,v_A>$, $<p_2,v_B>$, and $dpath(V_B)$. The loop in step $(7)$ is computed analogously by considering an ascending path.


\subsection{Orientable Surfaces With $\chi(\cdot)<0$ and Non-Empty Boundary}
\label{with boundary case}
In this section we present an algorithm to decompose a surface $M$ with $\chi(M)<0$ with boundary components. The algorithm follows almost similarly as before. The main difference is that we need to take care of the choice of the Morse function. We have two cases, $(1)$ The surface $M$ has exactly one boundary component $(2)$ The surface $M$ has more than one boundary component.\\

Let $M$ be a compact connected orientable surface with $\chi(M)<0$ and one boundary component $\Sigma$. We pick a point $p_{max}$ on the surface and construct a Morse function $f:M\longrightarrow[0,1]$ that satisfies the following conditions : 
\begin{enumerate}
\item 
$f^{-1}(\Sigma)=0$ and $f(x)>0$ for all $x$ in $M \backslash \Sigma$.
\item 
The point $f^{-1}(1)=x_{max}$ is a global maximum.
\item 
The function $f$ does not have any critical point of index $0$ or $2$ except for $p_{max}$. 
\end{enumerate}
The pants decomposition algorithm for a surface $M$ with  $\chi(M)<0$ and one boundary component $\Sigma$ algorithm goes now as follows : 
\begin{enumerate}
\item Compute the critical points of $f$ and put them in an ascending order. Let $p_1,p_2,...,p_n$ be the ordered critical points of $f$ and let $t_1,t_2,...,t_n$ be the corresponding critical values. By our choice of the Morse function, we have $index(p_i)=1$ for all $1\leq i \leq n-1$ and $index(p_{max})=2$.

\item For each $1\leq i < n-3$ let $c_i=\frac{t_i+t_{i+1}}{2}$. Define the set $C=\{c_i| 1 \leq i \leq 2g-3\}$. In other words, the set $C$ is a set of ordered regular values for $f$ such that there is exactly one critical value in the interval $[c_i,c_{i+1}]$. See Figure \ref{withboundary} for an example.

\item Cut the surface $M$ along the level sets $f^{-1}(c)$ for all $c\in C$.

\item Consider the manifold $M_{initial}:=M_{[-\epsilon,c_1]}$. By our choice of the Morse function $M_{initial}$ is a of type $(0,3)$.
\end{enumerate}

\begin{figure}[H]\centering
  \includegraphics[width=0.45\textwidth]{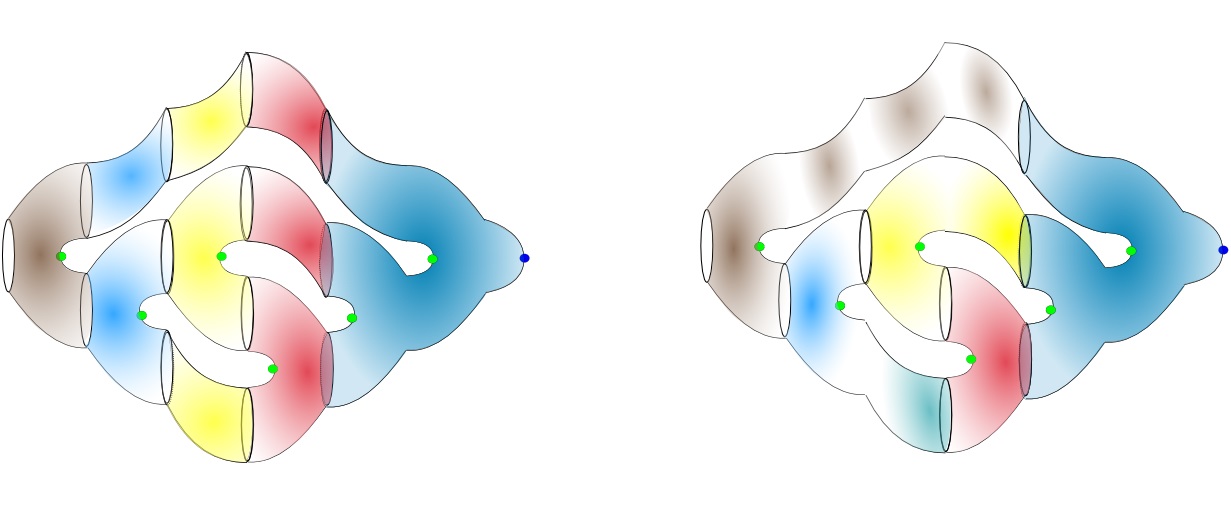}
   \put(-230,80){$(a)$}
   \put(-100,80){$(b)$}
    \caption{Illustration of the handle-based pants decomposition algorithm on a surface with a single boundary components : $(a)$ Cutting the surface along $c_1$, $c_2$. $c_3$ and $c_4$. $(b)$ Attaching cylinders at each level to the previous level and keeping the pants components. }
	\label{withboundary}
\end{figure}

The rest of the pants decomposition algorithm for a surface is similar to steps $(5)$, $(6)$ and $(7)$ of the algorithm in section \ref{alg1}.\\

Now we discuss the final case. Suppose that $M$ has boundary components $\Sigma_1,...,\Sigma_k$ where $k \geq 2$. Here we also need to construct a Morse function that serves our purpose. We need a Morse function $f:M\longrightarrow[0,1]$ that satisfies the following :
\begin{enumerate}
\item We choose one of the boundary component, say $\Sigma_1$, and we construct $f$ such that $f^{-1}(0)=\Sigma_1$ and $f(x)> 0$ for all $x\in M \backslash \Sigma_1$. 
\item   $f^{-1}(1)=\cup_{i=2}^{k} \Sigma_i$.
\item The function $f$ does not have any critical point of index $0$ or $2$.
\end{enumerate} 

As we did earlier, let $p_1,p_2,...,p_n$ be the ordered critical points of $f$ and let $t_1,t_2,...,t_n$ be the corresponding critical values of $f$. Define the values $c_i=\frac{t_i+t_{i+1}}{2}$ for all $1\leq t_i \leq n-1$. By our choice of the Morse function, the manifold  $M_{[c_i,c_{i+1}]}$ is homeomorphic to a finite disjoint union of one surface of type $(0,3)$ and multiple surfaces of type $(0,2)$. In particular $M_{[-\epsilon,c_1]}$ is a pant. The pants decomposition algorithm now is similar to the previous algorithm except that we do not trace any loop in this case since all connected components after cutting along the regular values of $C$ are pants or cylinders. See Figure \ref{withboundaries} for an example.
\begin{figure}[h]\centering
  \includegraphics[width=0.5\textwidth]{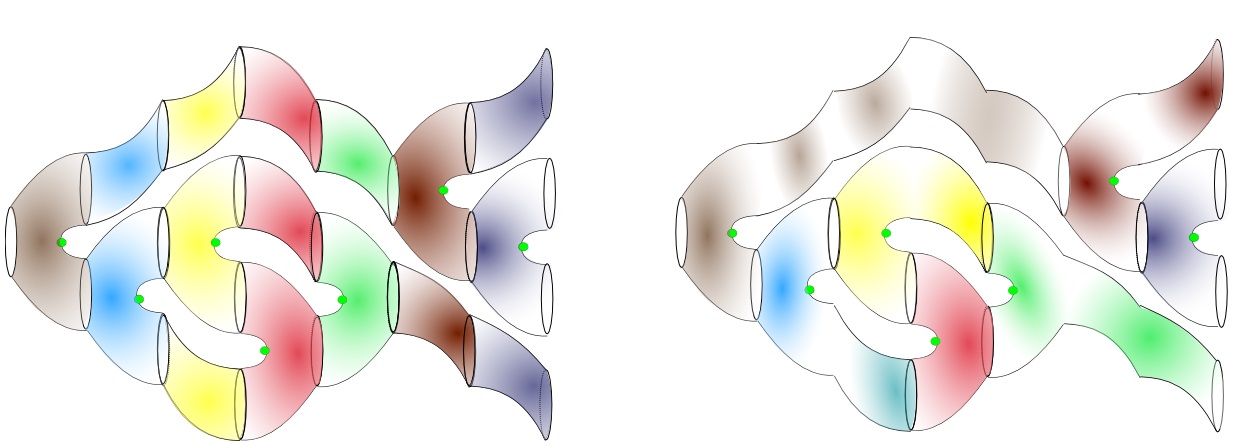}
   \put(-248,80){$(a)$}
   \put(-100,80){$(b)$}
   \caption{Illustration of the handle-based pants decomposition algorithm on a surface wit multiple boundary components : $(a)$ Cutting the surface along $c_1$, $c_2$. $c_3$, $c_4$, $c_5$ and $c_6$. $(b)$ Attaching cylinders at each level to the previous level and keeping the pants components. }
	\label{withboundaries}
\end{figure}

Figure \ref{alltog} shows multiple examples of pants decomposition of surfaces using this algorithm.

 \begin{figure}[h]\centering
  \includegraphics[width=0.6\textwidth]{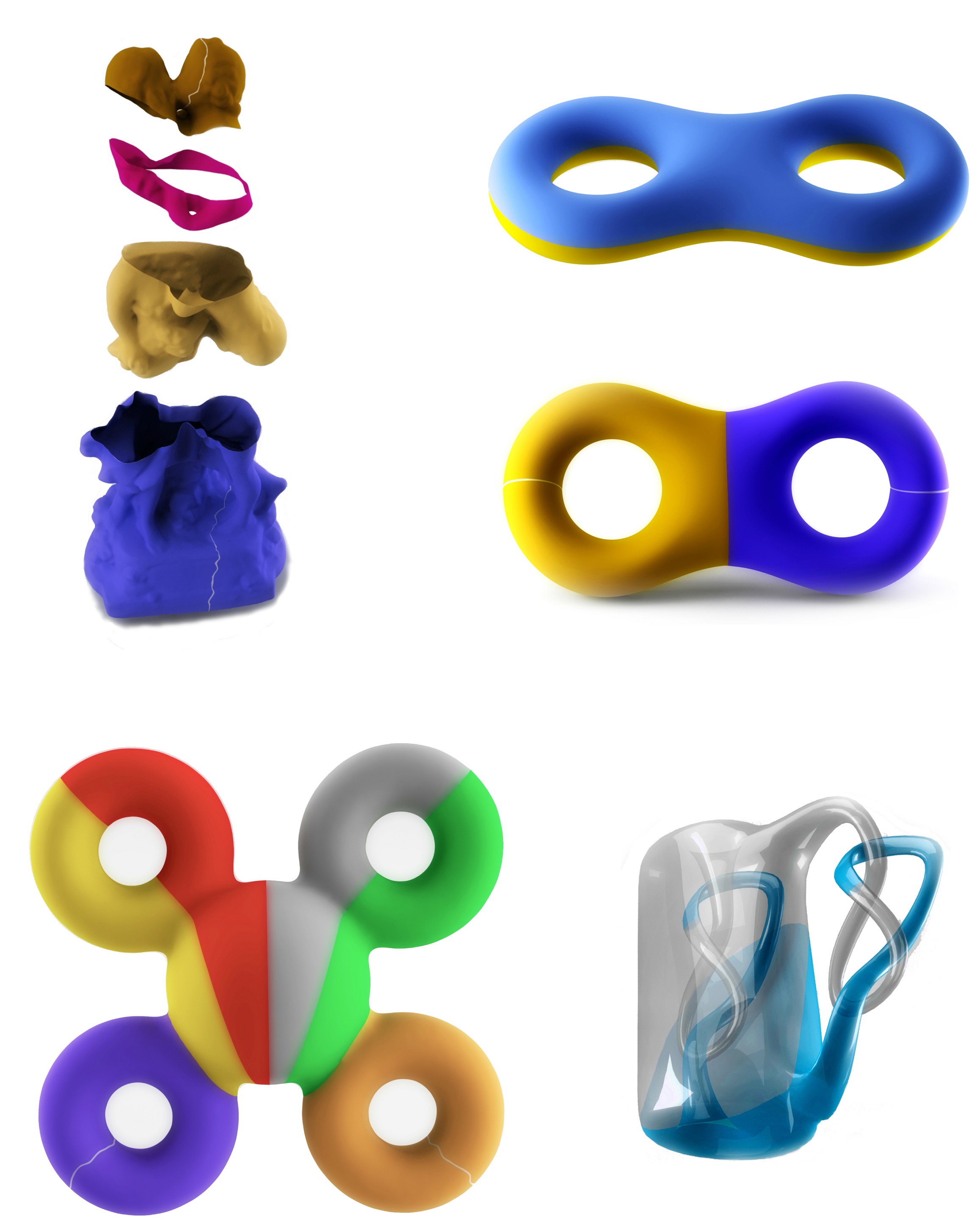}
  \caption{Pants decomposition using our Handle-based algorithm.}
  \label{alltog}
\end{figure}
\subsection{Dealing with Degenerate Cases}
\label{degenerate}

In this section we deal with the case when the index of a critical point point $p_i$  is $1$ and has multiplicity $m\geq 2$. We formulate our arguments in terms of Reeb graph for clarity.

\begin{enumerate}

\item If the point $p_1$ is a degenerate saddle point then the Reeb graph obtained from the quotient space of $M_{t_2+ \epsilon}$ appears as in Figure \ref{first case}.

\begin{figure}[h]\centering
  \includegraphics[width=0.2\textwidth]{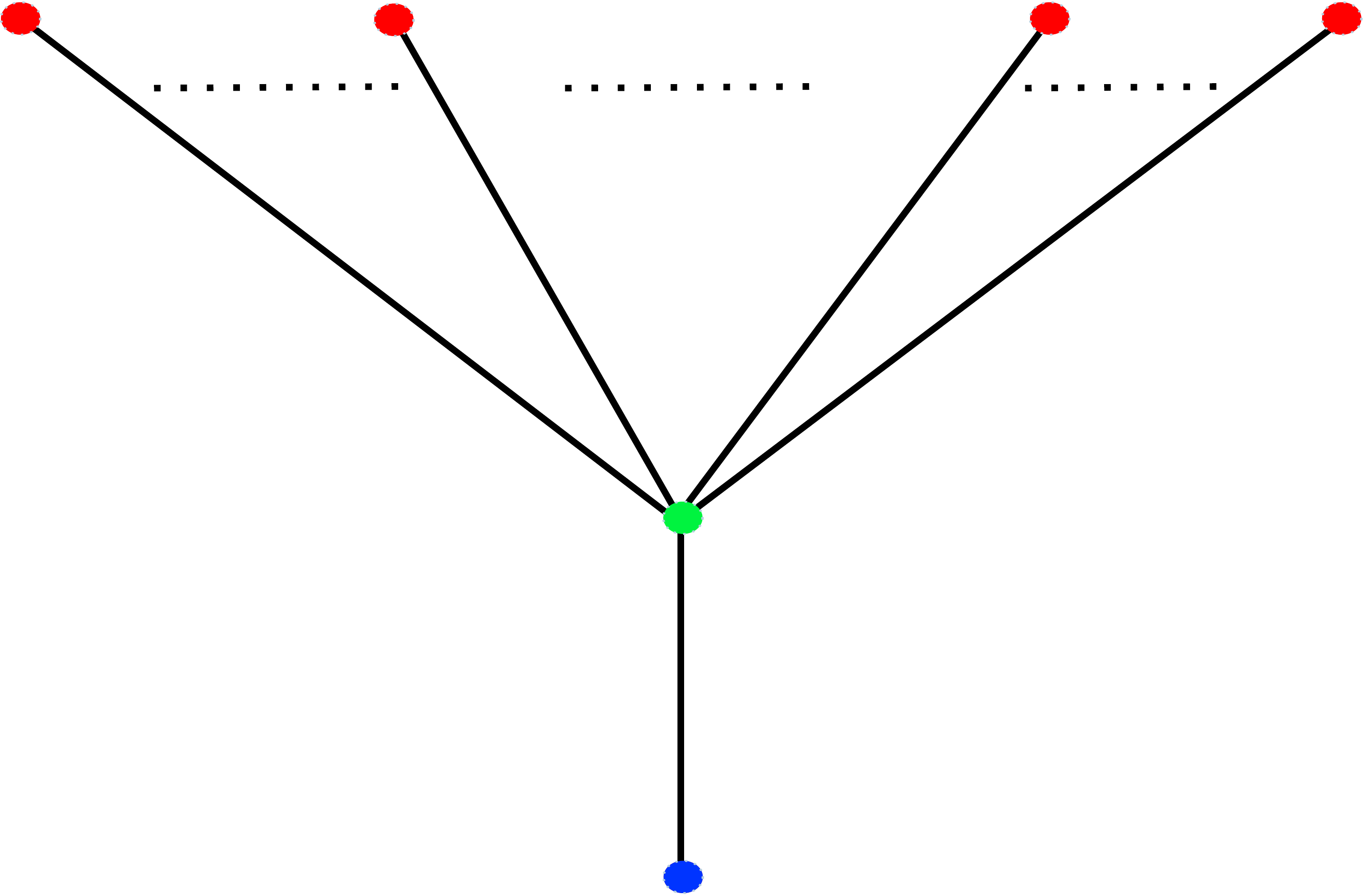}
   \caption{The Reeb graph obtained by restricting the PL scalar function on the manifold $M_{t_2+ \epsilon}$. The red point is the unique minimum of $f$ on $M$, the green point is
the degenerate saddle point, and the blue points represent the boundary circles of the surface $M_{t_2+ \epsilon}$.}
	\label{first case}
\end{figure}
In this case the surface $M_{t_2+ \epsilon}$ is of type $(0,b)$ where $b>3$ and we can decompose this surface into pants as described earlier.

\item The point $p_2$ is a simple saddle and the point $p_3$ is a degenerate one. In this case the Reeb graph obtained from the surface $M_{t_3+\epsilon}$ has two possibilities and they both appear in Figure \ref{second case}.

\begin{figure}[h]\centering
  \includegraphics[width=0.3\textwidth]{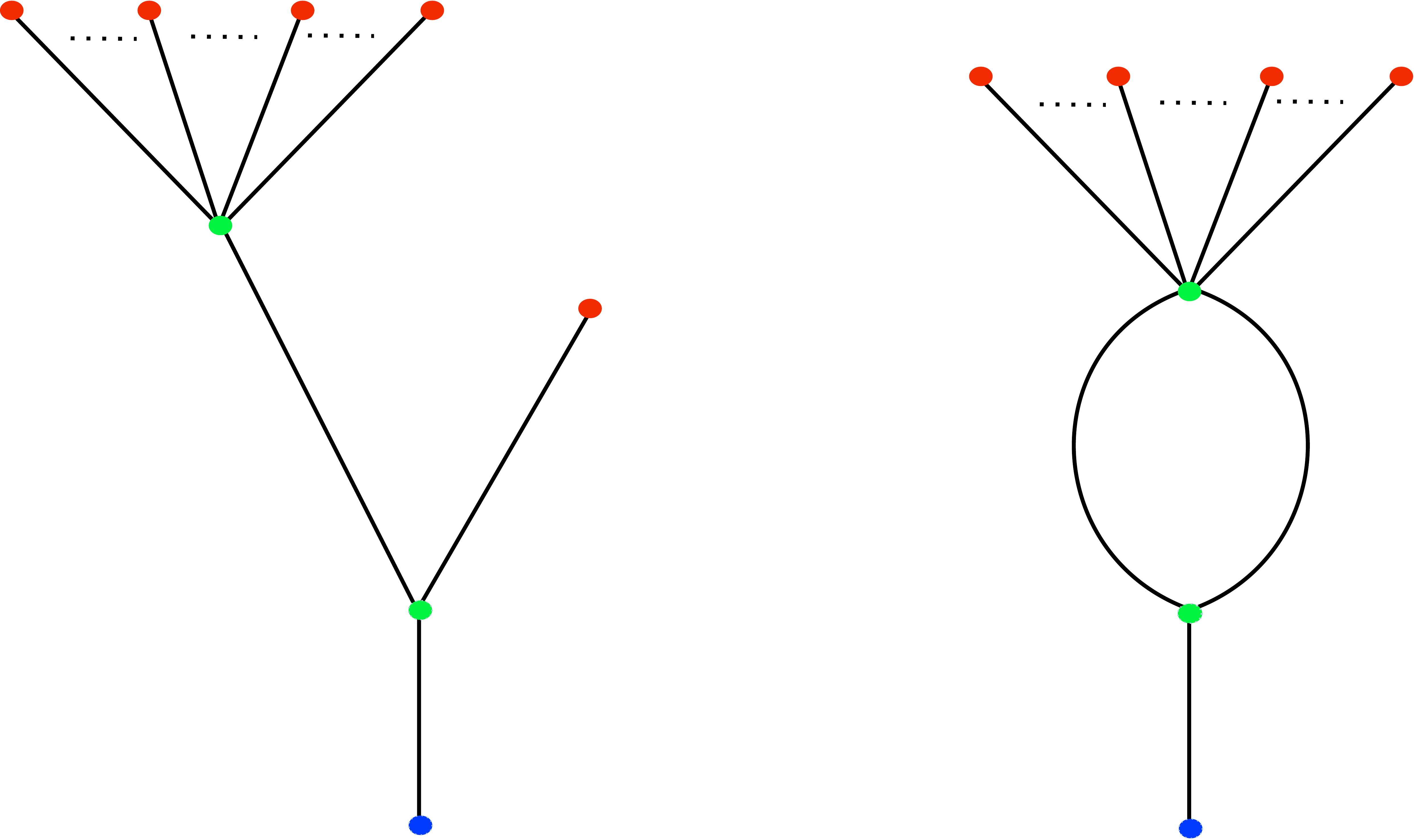}
   \caption{The two possible cases of the Reeb graph obtained by restricting the scalar function on the surface $M_{t_3+ \epsilon}$.}
	\label{second case}
\end{figure}
In the case when the Reeb graph appears as in left hand side of Figure \ref{second case} then $M_{t_3+ \epsilon}$ is homeomorphic to a surface of type $(0,b)$ where $b\geq 4 $. This surface can be decomposed into pants as we described earlier. The case when the Reeb graph appears as in the right hand side of Figure \ref{second case} then the $M_{t_3+ \epsilon}$ is surface of type $(1,b)$ where $b>2$. We can cut this surface along the simple saddle using descending path by tracing a loop from the simple saddle point the unique minimum using descending path as we described earlier into a sphere with $b+2$ boundary components. The latter can be also decomposed into pants as explained in earlier sections. Note that this case is an extension of the simple case we considered in Lemma \ref{main lemma}.

\item The case when $p_i$ is a degenerate saddle where $2<i<n-2$. In this case $f^{-1}(t_i-\epsilon,t_i-\epsilon)$, for a sufficiently $\epsilon$, is a finite disjoint union of multiple surfaces of type $(0,2)$ and one surface of of type $(0,b)$ where $b \geq 4$.  When we have this, we attach the cylinders just as before to previous pants and we apply Morse function-based algorithm again on the single remaining sphere with $b$ boundary components to decompose it into pants.
\end{enumerate}
\begin{remark}
We left out two cases, namely the cases when $p_{n-1}$ is degenerate and the case when $p_{n-1}$ is simple saddle and $p_{n-2}$ is a degenerate one. These two cases can be dealt with as cases (1) and (2) above.  
\end{remark}

\section{Reeb Graph-Based Pants Decomposition}
\label{alg2}
In this section we give pants decomposition algorithm using the Reeb graph of a Morse function. This algorithm is implicit in the work of \cite{hatcher1999pants}. We extend this algorithm to handle degenerate cases that show up in practice. The Reeb graph decomposition  algorithm has many advantages over the previous algorithm. The main advantage of this algorithm lies in the fact that it does not put any restriction on the choice of the Morse function. This allows us to choose a scalar function with better geometric properties. The second advantage is that choosing a cutting circle on a surface is much more flexible when using the structure of the Reeb graph than choosing the inverse image of a regular value of a Morse function.\\

 Let $M$ be a compact orientable (triangulated) surface, possibly with boundary, such that $\chi(M)<0$. Let $f$ be an arbitrary (PL) Morse function of $M$.  The Reeb graph-based pants decomposition algorithm of the surface $M$ and the Morse function $f$ goes as follows:

\begin{enumerate}
\item  Compute the Reeb graph $R(f)$ of $(M,f)$.
\item There are two types of $1$-valence nodes on the graph $R(f)$, the $1$-valence nodes that are a result of collapsing the boundary components of $M$ and the $1$-valence nodes that are coming from critical points of $f$ of index $0$ and $2$. We consider the graph $\overline{R(f)}$ obtained by taking the deformation retract of the graph $R(f)$ that leaves the edges coming from the boundary component without retraction. This step can be done by iteratively deleting the edges one of whose nodes has valence $1$ until there are no more such edges except the ones which have $1$-valence nodes originating from boundary components. See Figure \ref{reeb graph algorithm} step $(3)$ for an example of such retraction.
\item We remove all the nodes on the graph $\overline{R(f)}$ of valency $2$ and we combine the two edges that meet at such a node into one edge. We also denote the graph obtained from this step by $\overline{R(f)}$. Note that this graph is trivalent by construction. 
\item We select an interior point on every edge of the graph $\overline{R(f)}$ provided that neither one of the two nodes defining that edge has valency $1$. Note that the selection of the these point on the graph $\overline{R(f)}$ corresponds to partitioning the graph $\overline{R(f)}$ into a collection of small graphs each one of them is a vertex connected by small three arcs. See step $(5)$ in Figure \ref{reeb graph algorithm}. Each one of these small graphs corresponds to a pair of pants on the surface $M$.  
\item Each choice of an interior point induces a choice of a simple closed curve on the original surface $M$. The collection of all curves obtained in this way defines a pants decomposition of the surface $M$.
\end{enumerate} 

\begin{figure}[h]\centering
  \includegraphics[width=0.55\textwidth]{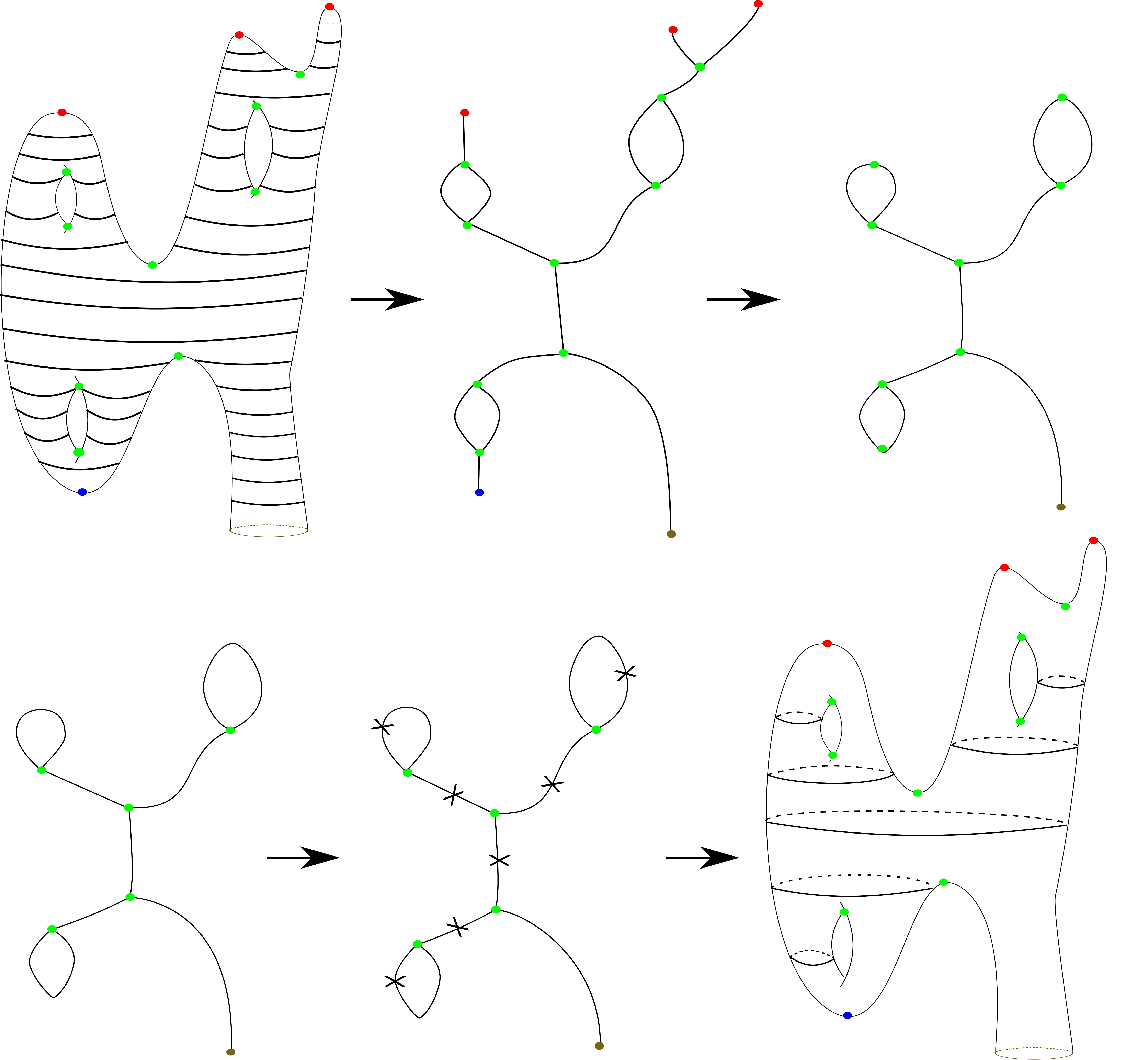}
   \put(-260,240){$(1)$}
   \put(-155,240){$(2)$}
   \put(-65,238){$(3)$}
   \put(-260,99){$(4)$}
   \put(-180,100){$(5)$}
   \put(-95,102){$(6)$}
   \caption{The steps of the Reeb graph-based pants decomposition algorithm}
	\label{reeb graph algorithm}
\end{figure}
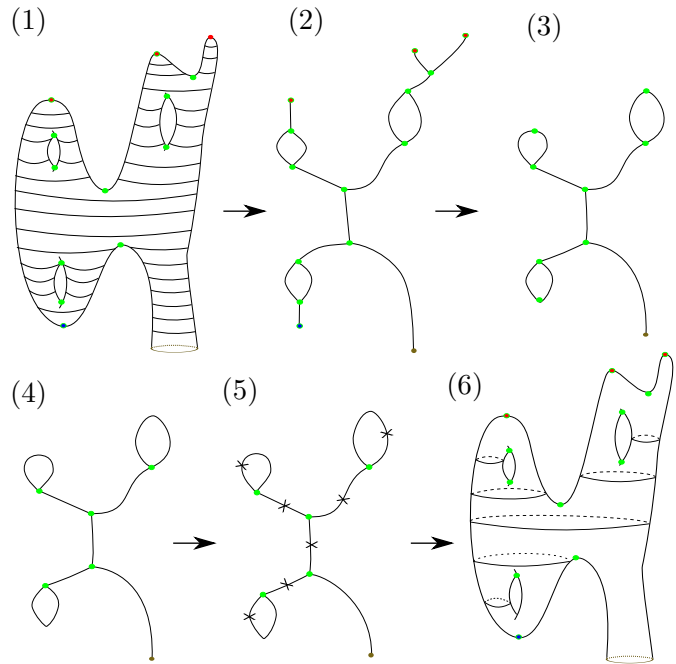

Figure \ref{reeb graph algorithm2} shows an application of this algorithm on some surfaces.

\begin{figure}[h]\centering
  \includegraphics[width=0.6\textwidth]{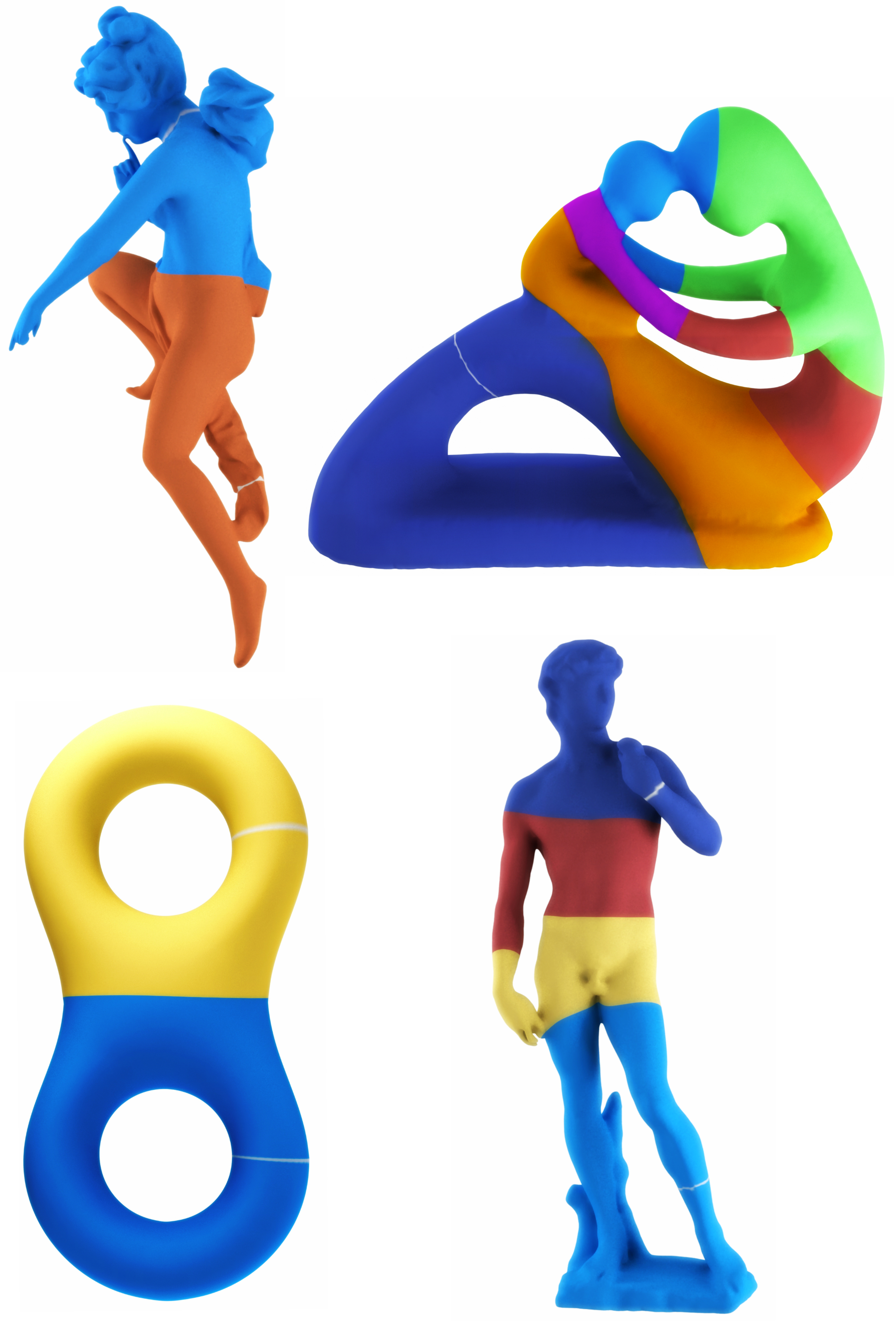}
   \caption{Pants decompositing using our Reeb graph-based algorithm}
	\label{reeb graph algorithm2}
\end{figure}

\subsection{Dealing with degenerate cases}
In the case when some critical points of index $1$ have multiplicity $m\geq 2$, we proceed as described in steps $1$ through $5$. However, the final result of the decomposition will no longer be a collection of pants but rather it will contain some surfaces of type $(0,b)$ where $b \geq 4$. For each single surface of these surfaces we apply one of the pants decomposition algorithms  again in order to decompose it into a collection of pants.

\section{Choosing an Appropriate Scalar Function}
\label{choosing scalar}
We need to construct a function that suits the algorithms that we have presented. For the algorithm given in section \ref{alg1} we need a scalar function that has one global minimum and one global maximum. This can be done by solving a Laplace equation on a mesh with Dirichlet boundary condition. A scalar $f$ function that satisfies the Laplace equation $\Delta f=0$ is called \textit{harmonic}. We are seeking here is a scalar function $f$ which satisfies the Laplace equation $\Delta f=0$ subject to the Dirichlet boundary conditions $f(v_i)=c_i$ for all $v_i \in  V_C$.  Here $V_C \subset V$ is a set of constrained vertices and $c_i$ are known scalar values providing the boundary conditions. This system has a unique solution provided $|V_C|\geq 2$. Furthermore, the solution for such a system has an important property that it has no local extrema other than the constrained vertices. This property of harmonic functions is usually called the maximum principle~\cite{rosenberg1997laplacian}. Designing such a function is possible in practice. Recall that on triangulated mesh $M$ the standard discretization for the Laplacian operator at a vertex $v_i$ is given by :
\begin{equation*}
\Delta f(v_i)=\sum_{[v_i,v_j]\in M}w_{ij}(f(v_j)-f(v_i)), 
\end{equation*}
where $w_{ij}$ is a scalar weight assigned to the edge $[v_i,v_j]$ such that $\sum_{[v_i,v_j]\in M}w_{,j}=1$. Choosing the weights $w_{ij}$ such that $w_{ij} >0 $ for all edges $[v_i,v_j]$ guarantees the solution of the Laplace equation has no local extrema other than at constrained vertices $V_C$ ~\cite{floater2003mean}. These conditions are satisfied by the\textit{ mean value weights}:
\begin{equation*}
w_{ij}=\frac{tan(\theta_{ij}/2)+tan(\phi_{ij}/2) }{||v_j-v_i ||},
\end{equation*}
where the angles $\theta_{ij}$ and $\phi_{ij}$ are the angles on either sides of the edge $[v_i,v_j]$ at the vertex $v_i$. Mean value weights are used to approximate harmonic map and they have the advantage that they are always non-negative which prevents any introduction of extrema on non-constrained vertices in the solution of the Laplace equation specified above. On the other hand, the cotangent weights may become negative in presence of oblique triangles and this can produce local extrema on non-constrained vertices. In this context see also \cite{wardetzky2007discrete} for various discretizations of the Laplace-Beltrami operator and their properties. In our Morse function-based algorithm for pants decomposition we needed a scalar function that has precisely one global minimum and one global maximum. Hence, the constrained vertices $V_C$ is chosen to have exactly two vertices $V_C=\{v_{min},v_{max}\}$ such that $f(v_{min})<f(v_{max})$. This choice will guarantee that the solution $f$ has a single minimum at $v_{min}$ and a single maximum at $v_{max}$. Hence, every other critical point for $f$ must be a saddle point which is a requirement for \ref{alg1}. Similarly, a harmonic scalar function can be used to obtain the pants decomposition algorithm \ref{with boundary case} for a surface with multiple boundary components.  \\

Even though our second algorithm works on a generic Morse function, we choose a function that captures the geometry and the symmetry aspects of the mesh. This was not possible in the previous algorithm due to the restriction of the input function. Our choice for scalar functions were made based on the object itself. For organic objects like the human body we choose the Poisson field \cite{dong2005harmonic}. On the other hand, for objects with some symmetry, we found that isometry invariant scalar functions \cite{sun2009concise,kim2010mobius,wang2014spectral} and multi-source heat kernel maps \cite{hajij2015constructing} give the best results.

\section{Experiments}
\subsection{Run-time comparison}
We ran our experiments on a $3.70$ GHz AMD(R) A6-6300 with 10.0 Gb memory.  We implement all the algorithms presented in Table \ref{Table} using C++ on a Windows platform. The algorithms were tested on different models and compared their run-time with a
publicly available software running the 
algorithm of \cite{li2009surface}. Table \ref{Table} shows that time comparison between these algorithms. The running times that are provided in the table includes the time needed for computing the scalar functions needed as an input for in out algorithms and exclude the time needed to compute the homology generators needed for the algorithm of \cite{li2009surface}. We did not include the pants decomposition algorithm presented in \cite{zhang2012optimizing} since it is not different fundamentally from the one in \cite{li2009surface}. The main purpose of the framework in \cite{li2009surface} is to enumerate different pants decomposition classes starting from an initial pants decomposition which is computed using the method given in \cite{li2009surface}. The algorithms that we presented here perform much better than the algorithm in \cite{li2009surface}. In particular, the Reeb graph algorithm gives us the best time efficiency.

\begin{center}
\begin{table}
\resizebox{\columnwidth}{!}{
\resizebox{0.04\textwidth}{!}{
\begin{tabular}{ |c|c|c|c|c|c| } 
\hline

Model& Vertices&Topology & Alg. $1$ time &  Alg. $2$ time. & Alg. \cite {li2009surface} time \\
\hline 
Eight &3070& $G=2$ &0.588 &0.322 &11.74\\ 
David & 26138 & $G=3$ &24.637&5.593 &64.43\\
$4$-Torus &10401 & $G=4$ &4.1495 &1.815 &4.52\\
Vase &10014 & $G=2$ & 2.793&1.224&4.12\\
Greek &43177 & $G=4$ &90.426 &10.567&275.2\\
Topology &6616 & $G=13$ &7.098&2.827 &Fail\\  
Knotty &5382 & $G=2$ &1.244 &0.606 &Fail\\ 
\hline
\end{tabular}}}\caption{Run-time in seconds; "Fail" means that the software crashes on the input mesh. The running time is in seconds.}\label{Table}
\end{table}
\end{center}

We mentioned in \ref{choosing scalar} that our choice of the scalar function relies on the geometry that we deal with. In Table \ref{Table2} we list our choices for the scalar functions that we used in our experimentation in the Reeb graph-based pants decomposition algorithm. We also list in the same table other scalar functions that we found potentially useful and could be used on the same model to produce similar results. In our Handle-based pants decomposition algorithm all scalar fields are Harmonic fields by our remark in Section \ref{choosing scalar}.

\begin{table}[h]
\resizebox{\columnwidth}{!}{
\resizebox{0.04\textwidth}{!}{
\begin{tabular}{ |c|c|c| }
\hline
Model&  Our choice of the scalar field& Other appropriate scalar fields   \\
\hline 
Eight &Laplacian Eigenfunctions & Harmonic functions,  \\ 
David &Poisson fields &  Harmonic functions  \\
$4$-Torus &Multi-source heat kernel maps& Isometry invariant fields \\
Vase &Harmonic functions& Poisson fields  \\
Greek &Harmonic functions& Poisson fields  \\
Topology &Isometry invariant fields& Poisson fields \\  
Knotty &Poisson fields & Harmonic functions \\
\hline
\end{tabular}}}\caption{The scalar fields used on each model in our experimentation for the Reeb graph-based algorithm.} \label{Table2} 
\end{table}
Figure \ref{two} shows examples of the scalar functions used in the input for our algorithms. The Figure also shows the critical points of the scalar functions.

\begin{figure}[h]\centering
  \includegraphics[width=0.5\textwidth]{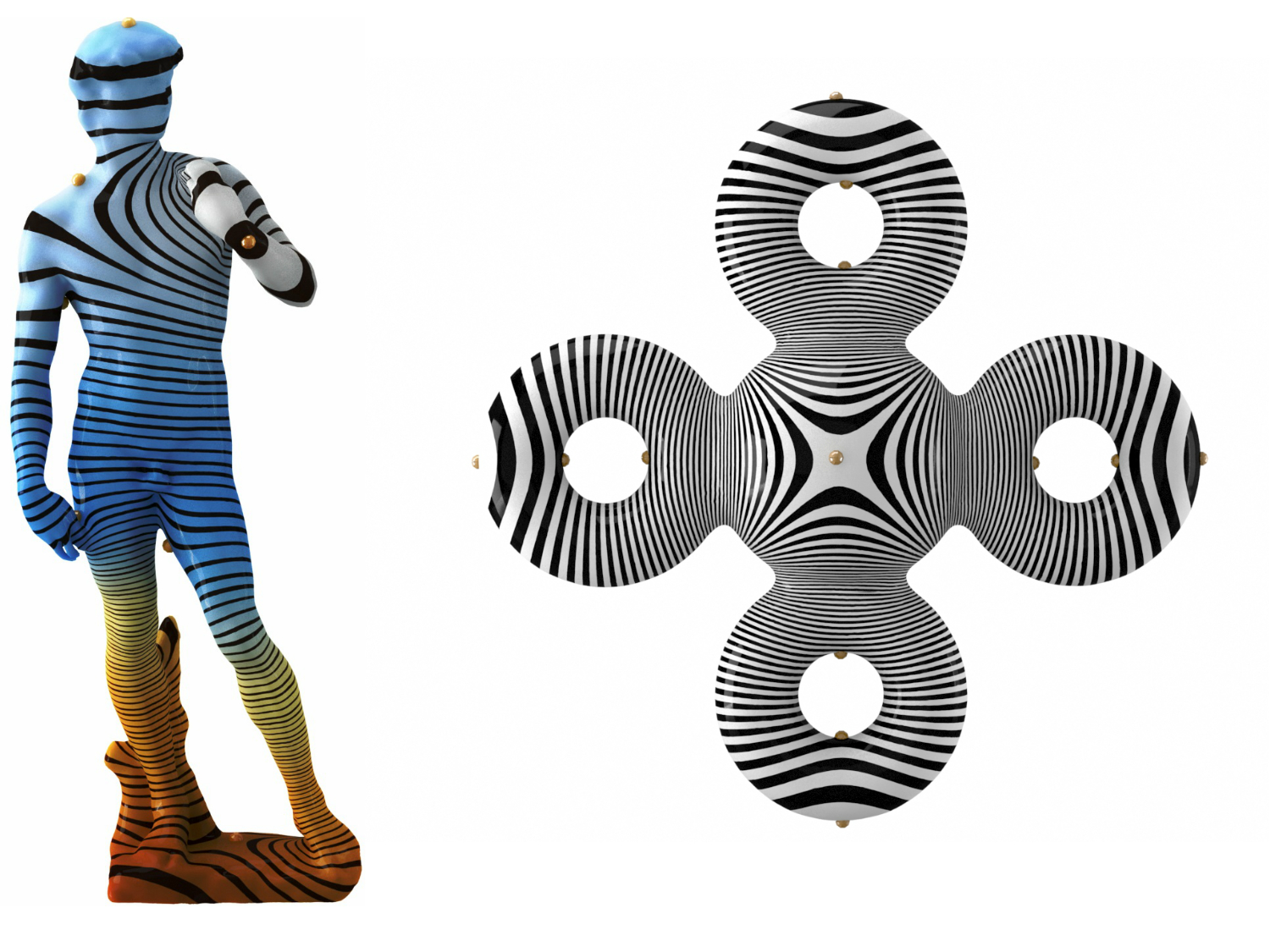}
   \caption{Two of the scalar functions used in the input for our algorithms. One the David model on the left a Poisson field is produced and on a multi-source heat kernel scalar field on the $4$-Torus model on the right.  }
	\label{two}
\end{figure}
\subsection{Imperfect input}
Just to test how robust are our algorithms against noise, we ran
our algorithms on surfaces corrupted with deliberate noise. 
We applied random displacements to the coordinates of the input mesh and tested the results. We applied noise amplitude up to $15\%$ 
of the bounding box length for the input mesh. Both algorithms 
maintained correct and meaningful outputs. 
See Figure \ref{noise} for some examples.
\begin{figure}[h]\centering
  \includegraphics[width=0.5\textwidth]{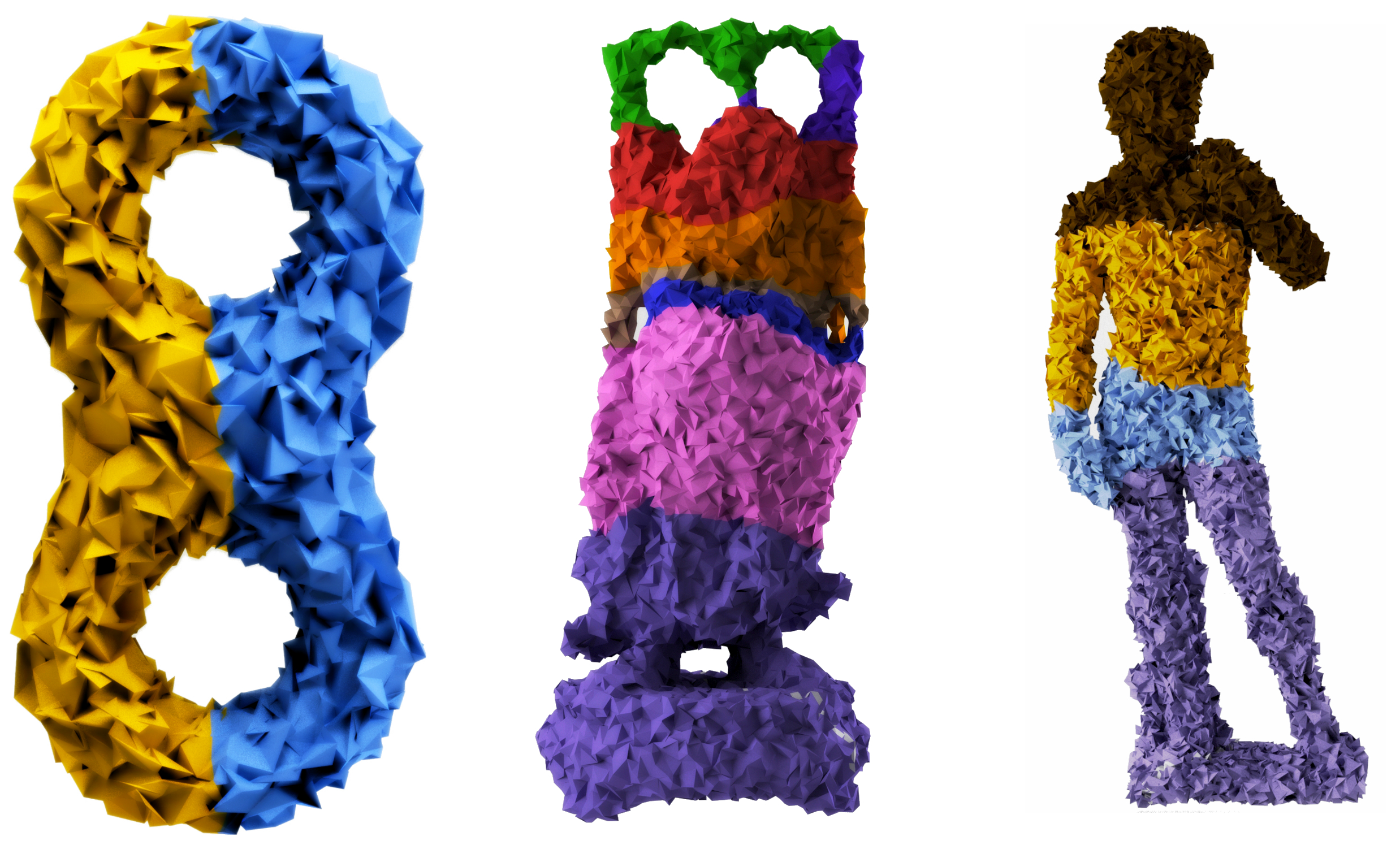}
   \caption{Pants decomposition for a noisy surface.}
	\label{noise}
\end{figure}

%

\section{Conclusion}
Morse theory is a powerful mathematical tool that uses local differential properties of a manifold to infer its global topological properties. 
Using a PL version of the theory, we have given two algorithms to decompose a surface mesh into pants. The algorithms we proposed here remove the constrains required by earlier algorithms and the run time much faster than the existing state-of-the-art method. \\

The choice of a Morse function has an impact on the output of the pants decomposition algorithms that we propose here. We have proposed some scheme to compute 
functions suitable for our pants decomposition algorithms. These scalar functions satisfies certain \textit{desirable} geometric properties. By the term \textit{desirable} we mean one or more of the following properties

\begin{enumerate}
\item The isolines of the scalar function are shape-aware in the sense that they follow one of the principal directions of the surface.

\item The critical points of the scalar function coincide with feature or the symmetry points on the surface.

\item If the surface has some sort of symmetry then the scalar field also inherits the symmetry of the surface.

\item Minimal user input.

\end{enumerate}

 However, are there alternative schemes which give better pants with some desirable geometric properties?
We intend to address this issue in future.





\bibliographystyle{plain}
\bibliography{sample}




\end{document}